\documentclass{imsart}

\usepackage[margin=1in]{geometry}
\usepackage[toc]{appendix}
\usepackage{graphicx}
\usepackage{amsmath,amssymb,amsthm,bm,latexsym}
\usepackage{graphicx,psfrag,epsf}
\usepackage[dvipsnames]{xcolor}
\usepackage{enumerate}
\usepackage[hyphens]{url}
\usepackage{booktabs}
\usepackage{multirow}
\usepackage{hyperref}
\usepackage{mathtools}
\usepackage{float}
\usepackage{subcaption}
\usepackage{mathrsfs,pstricks,multicol}
\usepackage[utf8]{inputenc}
\usepackage{listings}
\usepackage{cleveref}
\usepackage{placeins}

\usepackage{longtable}
\usepackage{natbib}
\usepackage{dsfont}
\usepackage{float}
\usepackage{array}
\newcolumntype{P}[1]{>{\centering\arraybackslash}p{#1}}
\usepackage{tabularx}

\usepackage{algpseudocode}
\usepackage[linesnumbered,ruled]{algorithm2e}
\hypersetup{colorlinks=true,citecolor=blue,urlcolor=purple,
	pdfpagemode=FullScreen,linktocpage=true}

\def\R{\mathbb{R}}

\def\T{\mathcal{T}}

\newcommand{\tminus}{\mathcal{T}_{-}}
\newcommand{\tplus}{\mathcal{T}_{+}}

\newcommand{\abs}[1]{\left\lvert #1 \right\rvert}
\newcommand{\norm}[1]{\left\| #1 \right\|}

\newcommand{\gauss}{\mathcal{N}}

\newcommand{\ind}{\mathbb{I}}

\def\argmax{\mathop{\rm argmax}}

\newtheoremstyle{general}
{3mm} 
{3mm} 
{\it} 
{} 
{\bfseries} 
{.} 
{.5em} 
{} 

\theoremstyle{general}

\newtheorem{lemma}{Lemma}
\newtheorem{theorem}{Theorem}
\newtheorem{corollary}{Corollary}
\newtheorem{assumption}{Assumption}

\makeatletter
\renewenvironment{proof}[1][\proofname]{\par
    \pushQED{\qed}%
    \normalfont \topsep6\p@\@plus6\p@\relax
    \trivlist
    \item\relax{
        \bfseries
        #1\@addpunct{.}}\hspace\labelsep\ignorespaces
    }{%
     \popQED\endtrivlist\@endpefalse
     }
\makeatother

\endlocaldefs

\begin{document}

\begin{frontmatter}
\title{Nonparametric method of structural break detection in stochastic time series regression model}
\runtitle{ Nonparametric break detection in time series}

\begin{aug}
\author[IIMK]{\fnms{Archi} \snm{Roy}\,\ead[label=e1]{archiroy@iimk.ac.in}},
\author[IISERP]{\fnms{Moumanti} \snm{Podder}\,\ead[label=e2]{moumanti@iiserpune.ac.in}},
\author[IIMB]{\fnms{Soudeep} \snm{Deb}\,\ead[label=e3]{soudeep@iimb.ac.in}}

\address[IIMK]{Indian Institute of Management Kozhikode, IIMK Campus P. O, Kunnamangalam, Kerala 673570, India.}
\address[IISERP]{Indian Institute of Science Education and Research, Dr Homi Bhabha Rd, Pune, Maharashtra 411008, India.}
\address[IIMB]{Indian Institute of Management Bangalore, Bannerghatta Main Rd, Bangalore, Karnataka 560076, India.}

\runauthor{Roy, Podder and Deb}
\end{aug}

\begin{abstract}
We propose a novel nonparametric test to detect structural breaks in the conditional mean and/or variance of a time series. Our method does not assume any specific parametric form for the dependence structure of the regressor, the time series model, or the distribution of the noise. This flexibility allows our algorithm to be applicable to a wide range of framework. We further apply the proposed test to accurately localize the changepoints and establish theoretical guarantees showing that the estimated structural breaks are consistent, meaning they lie sufficiently close to the true breakpoints when a sufficiently large sample is available. The effectiveness of the proposed algorithm is demonstrated through an extensive simulation study encompassing a diverse range of time series structures, including light, moderately heavy, and heavy tailed distributions. We also show a real-life example, where an application to Bitcoin prices and Google search volume illustrates how the procedure can identify changes in the conditional relationship between market attention and price dynamics.
\end{abstract}


\end{frontmatter}

\section{Introduction}

Detection of structural breaks is a crucial part of analysis and forecasting of time series data, for which the topic has received much attention in diverse fields including finance \citep{andreou2009structural}, climate research \citep{beaulieu2012change}, and others. The general problem of structural break detection concerns the inference of a change in distributional characteristics for a set of time-ordered observations. The detection can be sequential \citep[see][and relevant references therein]{aue2024state} or retrospective \citep[see][for a brief review]{truong2020selective}. The latter stream of literature can be broadly classified into parametric and nonparametric detection procedures. The parametric methodologies assume the underlying data generating process to follow a known distribution or a functional form. For example, \cite{robbins2020fully,bai2023multiple} developed methods of detecting changes in linear regression model, \cite{robbins2016general} developed a test for structural breaks in a linear regression model with autoregressive moving average (ARMA) residuals using a Wald test statistic, while more recently \cite{andersen2022testing} proposed a sup-Wald type test for structural breaks under vector autoregressive dynamics. On the other hand, sequential Monte Carlo methods have been developed \citep{chen2011detection} for detecting breaks in generalized autoregressive conditional heteroskedastic (GARCH) models and stochastic volatility models, both of which are common techniques to deal with financial datasets. A major issue with such parametric procedures is their specific assumptions about the structure of the underlying process, which may be impractical in real-life applications.

Parallelly, there have been multiple attempts in the extant literature to detect structural breaks avoiding the imposition of any parametric structure on the concerned time series. For example, \cite{liu2013change} detected breaks in the path of a time series by looking at the f-divergence  measure between the likelihood of the data in consecutive scanning windows in either side of potential break-points.  \cite{haynes2017computationally} proposed a cost function based on an initial segmentation of the data where the position of breaks are determined as the solution of the optimal segmentation that maximizes the log-likelihood obtained from the empirical distribution function. Similar work in the same direction were done by \cite{diop2023general}. In parallel, \cite{zhang2018unsupervised} proposed a self-normalized test procedure based on a cumulative sum (CUSUM) type test statistic to detect breaks in mean or other distributional properties of the time series. \cite{sundararajan2018nonparametric} proposed a test of detection of structural breaks in the covariance structure of multivariate time series utilizing the Euclidean difference in the spectral density matrices. More recently, \cite{fu2023multiple} proposed a methodology for detecting structural breaks in distributional properties of a time series with the help of the empirical distribution functions, and \cite{casini2024change} proposed a nonparametric algorithm for detecting structural breaks in the spectral density of a locally stationary time series.

The extant literature suggest that a common approach in this stream of literature is to detect changes in some statistical property of a time-evolving variable by assuming it to be consistent between consecutive breaks. These works often focus on optimizing a well-defined loss function obtained from evaluating that certain quantity of interest in two segments of the time series sample \citep[see][among others]{gosmann2022sequential}. A contrasting procedure is to identify structural breaks through the maximization of a loss function computed based on the regression curves in the different segments of the data. In this technique, rather than assuming a property to be constant between two consecutive breaks, it is considered to be an arbitrary smooth curve between breaks. Our proposed methodology aligns with this setup in particular. Many interesting developments in this direction have been done using a fixed design model $Y_i=f(x_i)+\epsilon_i$, where $Y_i$'s are the observed response, $x_i$'s are deterministic terms, $\epsilon_i$'s are independent model errors and $f$ is an unknown smooth function. For example, \cite{xia2015jump} used local linear smoothing to estimate the piece-wise regression curve $f$ assuming different numbers and locations of structural breaks. The optimal positions of the structural breaks are chosen as the ones which minimize the jump information criteria for the model. In the domain of time series regression, \cite{vogt2015testing} worked with the model 
\begin{equation}
\label{eq:eq_lit_review}
    Y_t=\mu(X_t)+\epsilon_t,
\end{equation} 
and tested whether the shape of the mean regression function $\mu(.)$ stays the same for all time points $t$ in the sample space. The test statistic proposed in this paper is based on the kernel-based Euclidean distance between all possible pairs $\mu(u)$ and $\mu(v)$, for $u,v\in\R$. \cite{yang2020change} detected structural breaks under the same model using a CUSUM type statistic assuming the process $\{X_t\}$ to be $\alpha$-mixing. In another pertinent work, \cite{fu2019model} considered a modified version of \eqref{eq:eq_lit_review}, with the conditional mean being a smooth time-varying function $g_t(.)$. Their test was developed utilizing the Fourier transform of $Y_t$ using an instrumental variable to infer if $g_t$ is time invariant. More recently, \cite{cui2023state} considered detection of breaks in a special case of the model \eqref{eq:eq_lit_review} with $X_t=Y_{t-1}$, but with more general assumptions on the true properties of the $\mu(.)$ function. 

There have been fewer works in the time series location-scale model
\begin{equation}
    \label{eq:eq_lit_review_our_model}
    Y_t=\mu(X_t)+\sigma(X_t)\epsilon_t,
\end{equation}
which is an extension of \eqref{eq:eq_lit_review}. In an early study, \cite{wishart2010kink} considered the same model to estimate structural breaks in the first derivative i.e. the slope of the mean regression function $\mu(\cdot)$, by using an extension of the traditional zero-crossing technique developed by \cite{goldenshluger2006optimal}. Their setup allowed for long-range dependence in $\{X_t\}$.

Our focus in this paper is on a similar structure of stochastic regression model, under which we develop a novel structural break detection framework. The contribution of this paper is to provide a unified nonparametric testing framework for structural breaks in conditional mean and conditional variance functions in stochastic time-series regression, while allowing flexible temporal dependence in the covariate process through projection-based dependence measures. This distinguishes the proposed method from existing nonparametric regression changepoint procedures developed under independent or mixing assumptions, and from methods targeting changes in marginal distributions or spectral features rather than conditional regression functionals. Although the main theory is developed under the at-most-one-change (AMOC) assumption, the test remains applicable in multiple change-point scenarios using established sequential testing methods. We further highlight its practical relevance by introducing an efficient localization algorithm. The proposed approach allows us to detect large scale structural shifts in the data and is more robust to short-term anomalies. We make mild assumptions on the properties of these functions and allow for both short and long range dependence in the covariate. As illustrated in \Cref{sec:simulation}, despite making mild assumptions, our proposed procedure is not only computationally less intensive, but is also superior in terms of performance for heavy-tailed financial data. At this point, it is worth highlighting the recent work by \cite{li2024estimating}, which focuses on estimating structural breaks in general moment characteristics of locally stationary time series by testing for parameter instability in moment condition models. In contrast, our approach directly targets structural breaks in the first and second moments by detecting changes in the overall functional behavior of the corresponding moment functions, offering a distinctive and effective alternative.

The rest of the paper is organized in the following way. In \Cref{sec:methods}, we present the mathematical framework of the problem, while \Cref{sec:testing_procedure} contains the proposed methodology of structural break detection and the relevant asymptotic theory. We assess the empirical performance of the proposed methodology for various time series structures, along with moderate to heavy-tailed residual distributions, and present the findings in \Cref{sec:simulation}. Next, \Cref{sec:application} contains an application of the proposed method to cryptocurrency data. We conclude with some necessary and important remarks in \Cref{sec:conclusion}. In the interest of space and flow of the paper, detailed proofs for all results are deferred to the supplementary material.

\section{Mathematical framework}
\label{sec:methods}

\paragraph{Notations} Throughout this paper, $E(\cdot)$ and $V(\cdot)$ are used to indicate the expectation and variance of a random variable. For any matrix $A$, $A_{i,j}$ represents the $(i,j)^{th}$ element. The $p^{th}$ norm will be indicated by $\norm{.}_p$, while $\mathcal{L}^p$ is used for the space of all random variables with finite $p^{th}$ norm. We shall use $\xrightarrow{P}$ for convergence in probability and $\xrightarrow{d}$ for convergence in distribution. For any set $S\subset\mathbb{R}$, denote by $\mathcal{C}^p(S)$ the space of functions with the $q^{th}$ derivative bounded on $S$ for all integers $q\leqslant p$, and let $S(\delta)=\bigcup_{\omega\in S}\{x\mid \abs{x-\omega}\leqslant\delta\}$ denote the $\delta$-neighborhood of $S$. 

We consider the stochastic regression model 
\begin{equation}\label{eq:main-model}
    Y_t=\mu(X_t)+\sigma(X_t)\epsilon_t, \quad \text{for } t\in \T=\{1,2,..,n\},
\end{equation}
where $\{Y_t\}$ and $\{X_t\}$ are two real-valued time series observed over $n$ discrete time points, $\mu:\mathbb{R}\rightarrow\mathbb{R}$ is the conditional mean regression function, $\sigma^2:\mathbb{R}\rightarrow\mathbb{R}^+$ is the conditional variance function and $\{\epsilon_t\}$ is a real-valued independently and identically distributed (iid) random noise process with unit variance. We assume that there exist $\Lambda_1<\Lambda_2$ for which $\{X_t\}$ is almost surely bounded within $\mathcal{X}=[\Lambda_1,\Lambda_2]$. It is important to emphasize that this is a purely technical assumption and in theory the interval $[\Lambda_1,\Lambda_2]$ can be as large as possible. We begin with the at-most-one-change (AMOC) assumption, under which at most a single structural break is present in the time horizon $\T$. Now, the structural break testing problems we work with can be written as:
\begin{equation}
\label{eq:what_to_test_mu}
    \begin{split}
        H_0^\mu &: \mu(X_t)=\mu_0(X_t) \; \forall \; t\in\T\text{ vs } \\
        H_1^\mu &:\mu(X_t)=\mu_1(X_t)\ind\{t\leqslant \tau_0^\mu\} + \mu_2(X_t)\ind\{t > \tau_0^\mu\}
    \end{split}
\end{equation}
and 
\begin{equation}
\label{eq:what_to_test_sigma}
    \begin{split}
        H_0^\sigma &: \sigma(X_t)=\sigma_0(X_t) \; \forall \; t\in\T\text{ vs } \\
        H_1^\sigma &:\sigma(X_t)=\sigma_1(X_t)\ind\{t\leqslant \tau_0^\sigma\} + \sigma_2(X_t)\ind\{t > \tau_0^\sigma\}
    \end{split}
\end{equation}
where $\tau_0^\mu$ and $\tau_0^\sigma$ are the structural breaks in the conditional mean and variance functions, respectively. While the core theoretical results are developed under this assumption, the proposed test is not limited to a single changepoint. In \Cref{subsec:detection-of-structural-breaks}, we extend the methodology to the general multiple changepoint setting by applying the AMOC test sequentially, using a binary segmentation-type algorithm that iteratively partitions the time series and applies the test to each resulting segment until no further breaks are detected. In the above expressions, all mean and variance functions are assumed to be non-periodic in the respective ranges. Note that simultaneous testing of both \eqref{eq:what_to_test_mu} and \eqref{eq:what_to_test_sigma} can be carried out using a simple Bonferroni type correction.  At this stage, it is important to highlight that in all our theoretical derivations, $X_t$ is going to be assumed univariate for convenience, but the proposed methodology and the results apply directly even when $X_t$ is a vector-valued process of finite dimension, under the assumption of decorrelation among the coordinates. We consider the following framework in this regard, which is in line with a plethora of literature in this domain. 

\begin{assumption}
     Error process $\{\epsilon_t\}$ is independent of $\{X_t\}$. The consistency results for the conditional mean function $\mu(\cdot)$ require $\{\epsilon_t\}$ to have finite second moments, while those for the conditional variance function $\sigma^2(\cdot)$ additionally require finite fourth moments.
\end{assumption}
Although finite fourth moments may appear restrictive in the context of financial time series, which are often characterized by heavy tails, we note that this assumption is imposed only on the standardized noise process $\{\epsilon_t\}$ rather than on the observed returns $\{Y_t\}$ directly. Since the conditional variance function $\sigma^2(\cdot)$ absorbs the bulk of the tail behavior in the model \eqref{eq:main-model}, the implied marginal distribution of $\{Y_t\}$ can still exhibit heavy-tailed characteristics even when $\{\epsilon_t\}$ satisfies this moment condition. Indeed, as demonstrated in \Cref{sec:simulation}, the proposed procedure performs well under heavy-tailed noise distributions such as the $t$-distribution and power-law innovations, confirming the applicability of the method beyond the strict confines of this assumption.

\begin{assumption}
\label{Xtdependence}
The time series $\{X_t\}$ is stationary and admits a causal representation driven by a sequence of iid random variables $\{\eta_t\}_{t\in\mathbb{Z}}$. Let $\mathcal{F}_t=\sigma(\eta_s: s \leqslant t)$ denote the innovation history up to time $t$. We assume that $X_t=m(\ldots,\eta_{t-1},\eta_t)$, where $m$ is a measurable function. Let $f_X$ be the marginal density of $X_t$. For $i \geqslant 1$, let $f_{X_i\mid \mathcal{F}_{i-1}}(x)$ denote the conditional density of $X_i$ given $\mathcal{F}_{i-1}$, evaluated at $x$, and let $f'_{X_i\mid \mathcal{F}_{i-1}}(x)$ denote its derivative with respect to $x$. For an integrable random variable $Z$, define the martingale projection operator $P_\ell$ as $P_{\ell}(Z) = E(Z\mid \mathcal{F}_\ell) - E(Z\mid \mathcal{F}_{\ell-1})$, for $\ell\in\mathbb{Z}$. The temporal dependence structure of $\{X_t\}$ is then expressed through the quantity $\Xi_k$, which, for any $k\in\mathbb{N}$, is defined as
\begin{equation*}
\begin{split}
    \Xi_k = k\Theta_{2k}^{2} + \sum_{r=k}^{\infty} \left(\Theta_{k+r}-\Theta_r\right)^2,
    \quad
    \Theta_k=\sum_{i=1}^{k}\theta_i, \\
    \text{where } \theta_i = \sup_{x\in\mathbb{R}} \norm{P_0 f_{X_i\mid \mathcal{F}_{i-1}}(x)} + \sup_{x\in\mathbb{R}} \norm{P_0 f'_{X_i\mid \mathcal{F}_{i-1}}(x)}.
\end{split}
\end{equation*}
Note that the projection operator is applied to the conditional density process, which is random through the conditioning sigma-field, rather than to the marginal density $f_X$.
\end{assumption}

Hereafter, the joint filtration generated by $\{\eta_{t},\epsilon_{t}\}$ is denoted by $\mathcal{G}_{t}$. It should be noted that the term $\theta_i$ is a rough quantification of the contribution of $\eta_0$ in predicting $X_i$. Smaller values of $\theta_i$ would denote weaker dependence of $X_i$ on the previous observations of the process. Following this argument, we say that $\{X_t\}$ is a short-range dependent (SRD) process if $\Theta_{\infty} < \infty$. If $\Theta_\infty=\infty$, with sufficiently slow decay of $\theta_i$, we refer to the process as long-range dependent (LRD). Note that our proposed framework allows for both short and long--range dependence in the covariate $X$.  It can be shown that for SRD processes,  $\Xi_{n}=O(n)$, whereas for LRD processes the rate of decay of $\Xi_n$ is much slower. For example, if $\theta_i$ is of the form  $l(i)/i^\beta$, where $\beta>1/2$ and $l(.)$ is a slowly varying function, then we can write $\Xi_n=O\left(n^{3-2\beta}l^{2}(n)\right)$ or $\Xi_n=O\left(n\bar{l}^2(n)\right)$ where  $\bar{l}(n)=\sum_{i=1}^{n} \abs{l(i)}/i$. This follows from a straightforward application of Karamata's theorem. Interested readers may refer to \cite{wu2003empirical}. 

\begin{assumption}
\label{assumption_functions}
    For some $\delta>0$, $f_X(\cdot)$, $\mu(\cdot)$ and $\sigma(\cdot)$ are four-times differentiable functions in $\mathcal{X}(\delta)$. Also, $\inf_{x\in\mathcal{X}}\{f_X(x)\}>0$ and $\inf_{x\in\mathcal{X}}\{\sigma(x)\}>0$.
\end{assumption}

It is important to highlight that, under \Cref{assumption_functions}, the modeling framework in \eqref{eq:main-model}  covers a fairly large class of traditional time series models. For example, if $X_t=Y_{t-1}$ and $\sigma(\cdot)$ is a constant function, we get the class of AR models. As a special case, if we further put $\mu(x)=ax$  for some constant $a\in\R$, we obtain the linear AR process. Similarly, by setting $\mu(x)=a\max\{x,0\}+b\min\{x,0\}$ or $\mu(x)=a+be^{-cx^2}$, where $a,b,c \in \R$, we obtain the threshold AR and exponential AR processes, respectively. One can also obtain heavy-tailed data generating processes from \eqref{eq:main-model}, for example $\mu(x)=0$ and $\sigma(x)=\sqrt{a+bx^2}$, gives the traditional ARCH process. Further, letting $Y_t=X_{t+1}-X_t$ and assuming $\{\epsilon_t\}$ to be iid Gaussian, we obtain the discretized version of the continuous time stochastic diffusion model $dX_t=\mu(X_t)+\sigma(X_t)dW_t,$ where $\{W_t\}$ is a standard Brownian motion. As pointed out by \cite{fan2005selective}, this structure covers a wide class of financial models. Further, \Cref{Xtdependence} is in line with several linear and nonlinear time series processes \citep{wu2005nonlinear}. 

\section{Theory}
\label{sec:testing_procedure}

Our main objective is to develop a test-based procedure to detect a structural break in a conditional functional of the distribution of $Y_t$ given $X_t=x$. Below, we write $g(x)$ generically for the functional of interest. In this paper we primarily consider two cases: $g(x)=E(Y_t\mid X_t=x)=\mu(x)$ and $g(x)=\operatorname{Var}(Y_t\mid X_t=x)=\sigma^2(x)$. As a first step to develop the test, imagine that we divide the time domain $\T$ into two disjoint halves, denoted by $\tminus$ and $\tplus$ respectively. Under our asymptotic regime, the cardinality of both $\tminus$ and $\tplus$ will approach $\infty$. With a slight abuse of notation, we continue to denote the number of observations in each set by $n$. Assume the true functional characteristic $g(\cdot)$ to be $g_1(\cdot)$ on $\tminus$ and $g_2(\cdot)$ in $\tplus$. In the absence of a structural break, we should have $g_1(x)=g_2(x)$ for all $x\in\mathcal{X}$. However, if a structural break is indeed present, then $g_1(\cdot)$ and $g_2(\cdot)$ will exhibit significant difference over the range $\mathcal{X}$. For a fixed $x \in \mathcal{X}$, let $  g_{\text{diff}}(x) = g_1(x) - g_2(x),$ which, in presence of a structural break, should be large for some $x\in\mathcal{X}$. Therefore, a significantly large value of  $\sup_{x\in\mathcal{X}}\{\abs{g_{\text{diff}}(x)}\}$ indicates the presence of a structural break. Our procedure relies on this phenomenon. In Sections \ref{sec:point-wise_asymptotics} and \ref{sec:supremum_asymptotics}, we provide the estimates of both $g_{\text{diff}}(\cdot)$ and $\sup_{x\in\mathcal{X}}\{\abs{g_{\text{diff}}(x)}\}$ (taking $g$ to be conditional mean or the conditional variance function), and establish their asymptotic theory. We emphasize that while the theory directly leads to a test for detecting the presence of a structural break, the method can be extended to localize them as well. The localization algorithm is detailed in \Cref{subsec:detection-of-structural-breaks}. 

\subsection{Asymptotic theory for point-wise estimates}
\label{sec:point-wise_asymptotics}

In this subsection, we derive the asymptotic properties of $g_{\text{diff}}(x)$ for a fixed $x\in\mathcal{X}$, where $g$ is the conditional mean function $\mu(\cdot)$, or the conditional variance function $\sigma^2(\cdot)$. In line with \eqref{eq:what_to_test_mu}, let us use $\mu_1(\cdot)$ and $\mu_2(\cdot)$ to denote the mean function in the segments $\tminus$ and $\tplus$, respectively. For a fixed $x\in\mathcal{X}$, the point-wise disparity between these two functions is given by $\mu_{\text{diff}}(x) = \mu_{1}(x) - \mu_{2}(x)$. We consider a Nadaraya-Watson type estimator $  \widehat\mu_{\text{diff}}(x) = \widehat\mu_{1}(x) - \widehat\mu_{2}(x)$, with
\begin{equation*}
    \widehat \mu_{1}(x) = \frac{1}{nb_n \widehat f_X(x)} \sum_{\tminus}{Y_{t} K\left(\frac{x - X_{t}}{b_n}\right)}, \;
    \widehat \mu_{2}(x) = \frac{1}{nb_n \widehat f_X(x)} \sum_{\tplus}{Y_{t} K\left(\frac{x - X_{t}}{b_n}\right)},
\end{equation*}
where $\widehat f_{X}(x) = (nb_n)^{-1} \sum_{\T}{K\left((x - X_{t})/b_n\right)}$ is the estimated density of the covariate process $\{X_t\}$, $K(\cdot)$ is an appropriately chosen kernel function, and $b_n=b(n)$ is a bandwidth sequence satisfying the following assumption.

\begin{assumption}
\label{asmp:kernel-bw-conditions}
    The kernel function $K(.)$ is symmetric, bounded, has bounded derivative and bounded support $[-1,1]$. The bandwidth sequence $b_n=b(n)$ satisfies $b_n\rightarrow 0$ and  $nb_n\rightarrow \infty$.
\end{assumption}

Akin to above, denote the conditional variance function $\sigma^2(\cdot)$ in the segments $\tminus$ and $\tplus$ as $\sigma_1^2(\cdot)$ and $\sigma_2^2(\cdot)$ respectively. For a fixed $x\in\mathcal{X}$, the point-wise difference $\sigma^2_{\text{diff}}(x)=\sigma_1^2(x)-\sigma_2^2(x)$ is estimated as  $\widehat\sigma^2_{\text{diff}}(x) = \widehat\sigma_1^2(x)-\widehat\sigma_2^2(x), $ where
\begin{equation*}
\begin{split}
    \widehat\sigma_1^2(x) &= \frac{1}{nb_n\widehat f_X(x)}\sum_{\tminus}{(Y_{t}-\widehat{\mu}_{1}(X_{t}))^{2}K\left(\frac{x-X_{t}}{b_n}\right)}, \\
    \widehat\sigma_2^2(x) &= \frac{1}{nb_n\widehat f_X(x)}\sum_{\tplus}{(Y_{t}-\widehat{\mu}_{2}(X_{t}))^{2}K\left(\frac{x-X_{t}}{b_n}\right)}.
\end{split}
\end{equation*}

Note that we have used the same bandwidth $b_n$ for the estimation of both the conditional mean and variance functions. One can also use a different bandwidth sequence $h_n$ for the estimation of variance. It does not affect the theoretical results provided that  $\lim_{n\rightarrow\infty}\abs{h_n/b_n}$ is bounded away from $0$ and $\infty$. More details on the choice of the bandwidth are provided in Section S1 of the supplementary materials. For the kernel function $K(\cdot)$, define $\phi(K)=\int_{\mathbb{R}}{(K(u))^{2}du}$ and $\psi(K)=\int_{\mathbb{R}} (u^{2}/2)K(u)du$. The following results describe the asymptotic point-wise behavior of $\widehat{\mu}_{\text{diff}}(\cdot)$, under specific conditions on whether the variance function should be assumed to be same for the entire time horizon or not. 

\begin{theorem}
\sloppy
\label{lem:theorem1}
Along with the assumptions in \Cref{sec:methods}, assume that the conditional variance function remains the same throughout the time domain $\T$, i.e., $\sigma_1^2(x)=\sigma_2^2(x)=\sigma^2(x)$ for all $x$, and suppose the bandwidth satisfies $ nb_n^9+\frac{1}{nb_n}+\Xi_{n}\left(\frac{b_n^3}{n}+\frac{1}{n^2}\right)\xrightarrow{n\rightarrow\infty}0$. Fix $x\in\mathcal{X}$ such that $f_X(x)>0,\sigma(x)>0$ and $f_X,\mu\in\mathcal{C}^4\left(x-\delta,x+\delta\right)$ for some $\delta>0$. Then, as $n\rightarrow\infty$,
\begin{equation*}
\label{eq:eqth1}
{\small
        \frac{\sqrt{nb_n\widehat f_X(x)}}{\sqrt{2\phi(K)\widehat\sigma^2(x)}} \bigg[ \widehat\mu_{\text{diff}}(x)- \mu_{\text{diff}}(x) -  \left(b_n^2\psi(K)\rho_{\mu_{1}}(x)-b_n^2\psi(K)\rho_{\mu_{2}}(x)\right)\bigg]
    \xrightarrow{d}\gauss\left(0,1\right),
    }
\end{equation*}
where $\widehat\sigma^2(x)$ is a consistent nonparametric estimate of the common conditional variance function, and the asymptotic bias of the estimate is defined through 
\begin{equation*}
        \rho_{\mu_1}(x) = \mu_1''(x)+2\mu_1'(x)\frac{f_{X}'(x)}{f_{X}(x)}, \;
        \rho_{\mu_2}(x) = \mu_2''(x)+2\mu_2'(x)\frac{f_{X}'(x)}{f_{X}(x)}.
\end{equation*}
\end{theorem} 

It is imperative to point out that, in practical applications, the true conditional mean functions $\mu_1(.)$ and $\mu_2(.)$ are not known. Hence, the bias terms involving $\rho_{\mu_1}(\cdot)$ and $\rho_{\mu_2}(\cdot)$ need to be estimated as well. Following \cite{wu2007inference}, we avoid this by utilizing a jackknife correction, which is equivalent to estimating $\widehat\mu_1(\cdot)$ and $\widehat\mu_2(\cdot)$ using the kernel $K^{*}(u)=2K(u)-K(u/\sqrt{2})/\sqrt{2}$. Note that $K^{*}$ has support $[-\sqrt{2},\sqrt{2}]$. Hereafter, the bias-corrected estimates will be denoted with $*$ above them, e.g., $\widehat\mu_1^*(\cdot)$. 

\begin{corollary}
\label{lem:corollary1}
    Under the assumptions stated in \Cref{lem:theorem1}, 
\begin{equation*}
    \label{eq:eqmodth1}
    \frac{\sqrt{nb_n\widehat f_X(x)}}{\sqrt{2\phi(K)\widehat\sigma^2(x)}}\left[\widehat\mu_{\text{diff}}^*(x)-\left(\mu_1(x)-\mu_2(x)\right)\right]\xrightarrow{d}\gauss\left(0,1\right),
\end{equation*}
where $\widehat\mu_{\text{diff}}^*(x)=\widehat\mu_1^*(x)-\widehat\mu_2^*(x)$ is the estimate of $\mu_{\text{diff}}(x)$ using the kernel $K^*$.
\end{corollary}

The proof of \Cref{lem:corollary1} is straightforward by noting that $\psi(K^*)=0$ for the modified kernel function. In the following corollary, we establish the point-wise behavior of the $\widehat\mu_{\text{diff}}^*(.)$ function under the setup where we allow for the possibility of a structural break in the conditional variance $\sigma^2(\cdot)$.

\begin{corollary}
\label{lem:corollary2}
    Assume that the conditional variance function may have different behavior in $\tminus$ and $\tplus$, and is estimated separately in the two segments as $\widehat\sigma_1^2(x)$ and $\widehat\sigma_2^2(x)$. Then, letting $\widehat S(x)=\widehat\sigma_1^2(x)+\widehat\sigma_2^2(x)$, under the same conditions specified in \Cref{lem:theorem1}, for a fixed $x\in\mathcal{X}$ as $n\rightarrow\infty$,
\begin{equation*}
\label{eq:eqprop2}
     \frac{\sqrt{nb_n\widehat f_X(x)}}{\sqrt{\phi(K)\widehat S(x)}}\left[\widehat\mu_{\text{diff}}^*(x)-\left(\mu_{1}(x)-\mu_{2}(x)\right)\right]\xrightarrow{d}\gauss\left(0,1\right).
 \end{equation*}
\end{corollary}

The proof of \Cref{lem:theorem1} involves expressing $\widehat\mu_{\text{diff}}(x)$ as a martingale difference sequence with respect to the filtration $\left\{\mathcal{F}_t\right\}$, followed by a straightforward application of martingale central limit theorem. Details can be found in Section S2 of the supplementary material. The proof of \Cref{lem:corollary2} is similar. 

As previously discussed, we similarly evaluate the disparity between the conditional variance functions in $\mathcal{T}_{-}$ and $\mathcal{T}_{+}$ using the expression $\widehat{\sigma}^2_{\text{diff}}(x) = \widehat{\sigma}_1^2(x) - \widehat{\sigma}_2^2(x)$. In this formulation, each of the estimates $\widehat{\sigma}_i^2(\cdot)$ depends on the respective conditional mean estimates $\widehat{\mu}_i(\cdot)$, for $i = 1, 2$. However, they inherently possess a bias of order $O(b_n^4)$. Although the asymptotic distribution of $\widehat{\sigma}^2_{\text{diff}}(x)$ can be derived using this bias, for improved performance of the test statistic, we incorporate the bias-corrected mean estimates $\widehat{\mu}_i^*(\cdot)$ in the estimation of the conditional variances. The revised estimates are therefore given as,
\begin{equation*}
\begin{split}
    \widehat\sigma_1^2(x) &=\frac{1}{nb_n\widehat f_X(x)}\sum_{\tminus}{(Y_{t}-\widehat{\mu}_{1}^*(X_{t}))^{2}K\left(\frac{x-X_{t}}{b_n}\right)}, \\
    \widehat\sigma_2^2(x) &=\frac{1}{nb_n\widehat f_X(x)}\sum_{\tplus}{(Y_{t}-\widehat{\mu}_{2}^*(X_{t}))^{2}K\left(\frac{x-X_{t}}{b_n}\right)}.
\end{split}
\label{eq:sigma1hat}
\end{equation*} 

Furthermore, it can be shown that both the above estimates are biased, with the bias being of the order $O(b_n^2)$. We therefore follow the jackknife type correction once again and use the modified kernel $K^*$ to obtain bias-corrected estimates $\widehat\sigma_i^{*^2}(x)$, for $i=1,2$, and define $  \widehat\sigma^{*^2}_{\text{diff}}(x)= \widehat\sigma_1^{*^2}(x)-\widehat\sigma_2^{*^2}(x)$. We next derive the point-wise asymptotic distribution of $\widehat\sigma^{*^2}_{\text{diff}}(.)$. 

\begin{theorem}
\label{lem:theorem2}
Along with the previously stated assumptions in \Cref{sec:methods}, consider the bandwidth condition $b_n^{\frac{3}{2}}\log n+\frac{1}{n^2b_n^5}+\frac{\Xi_n}{n^2}\xrightarrow{n\rightarrow\infty}0$.
Fix an $x\in\mathcal{X}$ such that $f_X(x)>0,\sigma(x)>0$ and $f_X,\mu,\sigma\in\mathcal{C}^4\left(x-\delta,x+\delta\right)$ for some $\delta>0$. If $\nu_\epsilon=E\left(\epsilon_0^4\right)-1$ and $\widehat S(x)$ is as defined in \Cref{lem:corollary2}, then as $n\rightarrow\infty$,
\begin{gather*}
\label{eq:eqprop3}
    \frac{\sqrt{nb_n\widehat f_X(x)}}{\nu_\epsilon \sqrt{\phi(K^*) \widehat S(x) }}\left[\widehat\sigma^{*^2}_{\text{diff}}(x)-\left(\sigma_1^2(x)-\sigma_2^2(x)\right)\right]\xrightarrow{d}\gauss(0,1).
\end{gather*} 
\end{theorem}

The proof of \Cref{lem:theorem2} follows in the same line as \Cref{lem:theorem1}, and the details are furnished in Section S2 of the supplement. Given that we do not impose any specific distributional assumptions on the random noise process, the moments of the distribution are not known and $\nu_\epsilon$ needs to be estimated. Denoting the standardized residuals from the model \eqref{eq:main-model} as $\{\widehat r_t\}$, one may replace the term $\nu_\epsilon$ in \Cref{lem:theorem2} by
\begin{equation*}
        \widehat{\nu_{\epsilon}}= \frac{\sum_{\tminus}{\widehat r_{1_t}^4 1_{\{X_t\in\mathcal{X}_-\}}}+\sum_{\tplus}{\widehat r_{2_t}^4} 1_{\{X_t\in\mathcal{X}_+\}}}{\sum_{\tminus}{1_{\{X_t\in\mathcal{X}_-\}}}+\sum_{\tplus}{1_{\{X_t\in\mathcal{X}_+\}}}}-1,
\end{equation*}
where $\widehat r_{1_t}=\frac{Y_t-\widehat\mu_1^*(X_t)}{\sqrt{\widehat\sigma_{1}^{*^2}(X_t)}}$ and $\widehat r_{2_t}=\frac{Y_t-\widehat\mu_2^*(X_t)}{\sqrt{\widehat\sigma_{2}^{*^2}(X_t)}}$ are the standardized residuals for \eqref{eq:eq_lit_review_our_model} estimated in the segments $\tminus$ and $\tplus$ respectively, and $\mathcal{X}_-$ and $\mathcal{X}_+$ are the ranges of the covariate $X$ in these two segments. It can be shown that under the previously mentioned conditions, $\widehat{\nu}_\epsilon$ is a consistent estimate for $\nu_\epsilon$.

\subsection{Test of presence of structural break}
\label{sec:supremum_asymptotics}

The asymptotic theory developed for the estimates $\widehat\mu_{\text{diff}}(\cdot)$ and $\widehat\sigma^2_{\text{diff}}(\cdot)$ enable us to assess the disparity between the conditional mean and variance functions in two disjoint halves of the data at specific covariate profile $x\in\mathcal{X}$. However, they do not provide sufficient information to determine whether the overall behavior of the functions differ in the two segments. Consider the specific example of the conditional mean function. As we are interested in the equality of the functional characteristic in $\tplus$ and $\tminus$, a change in the global behavior of $\mu(\cdot)$ may be estimated through the test statistic
\begin{equation}\label{eq:sup-mu1-mu2}
     \sup_{x\in\mathcal{X}}\left\{\abs{\widehat{\mu}^*_{\text{diff}}(x)}\right\}=\sup_{x\in\mathcal{X}}\left\{\abs{\widehat\mu_1^*(x)-\widehat\mu_2^*(x)}\right\}.
\end{equation}

Although it is theoretically possible to establish the asymptotic distribution of the above quantity over a continuous range, in practice, nonparametric estimates can only be computed point-wise. Keeping that in view, we present the asymptotic properties of \eqref{eq:sup-mu1-mu2} over a sufficiently fine grid of points in $\mathcal{X}$. As we shall present in the simulation studies later, the tests conducted using this discretization maintain satisfactory performance in practice. We define a sequence of partitions $\{\Pi_n\}$ of $\mathcal{X}$ as $\Pi_n=\left\{x_{t_j} \mid x_{t_j}=\Lambda_1+2jb_n,j=0,1,2,\hdots,m_n-1\right\},$ where $m_n=\left\lceil\frac{\Lambda_2-\Lambda_1}{2b_n}\right\rceil$. Note that the partitions $\{\Pi_n\}$ become dense in $\mathcal{X}$ as $n\rightarrow\infty$. It can be argued that under appropriate smoothness conditions on the true conditional mean function (as mentioned in \Cref{assumption_functions}), $\{\mu(x)\mid x\in\mathcal{X}\}$ can be well approximated by $\{\mu(x)\mid x\in\Pi_n\}$ for a sufficiently large value of $n$. We can therefore conclude that
\begin{equation*}
\label{eq:sup_mu}
    \sup_{x\in\mathcal{X}}\left\{\abs{\widehat{\mu}^*_{\text{diff}}(x)}\right\}\approx \lim_{n\rightarrow\infty}\left[\sup_{x\in\Pi_n}\left\{\abs{\widehat{\mu}^*_{\text{diff}}(x)}\right\}\right].
\end{equation*}

Our interest is in detecting the presence of structural break in the conditional mean function, and we can treat this as a testing of hypothesis problem of the form \eqref{eq:what_to_test_mu}, which can be rewritten in the following form:
\begin{equation}\label{eq:mu-exist-hypothesis}
    \begin{split}
        H_0: & \;  \text{ There is no structural break in the conditional mean}, \\
        H_1: & \;  \text{ There is a structural break in the conditional mean}.
    \end{split}
\end{equation}

In other words, the null hypothesis $H_0$ corresponds to the assumption $\mu_1(x) = \mu_2(x)$ for all $x\in\mathcal{X}$ whereas the alternative hypothesis $H_1$ points to the inequality of the two functions for at least one $x$. As mentioned before, we can use the statistic given in \eqref{eq:sup-mu1-mu2}, and use its null distribution to detect significant deviation from $H_0$. The following theorem provides the large sample behavior of the test statistic under the null hypothesis when the conditional variance function stays the same throughout the time domain $\T$.

\begin{theorem}
\label{lem:theorem3}
    Along with the previously stated assumptions in \Cref{sec:methods}, assume that the conditional variance function remains the same throughout the time domain $\T$, i.e., $\sigma_1^2(x)=\sigma_2^2(x)=\sigma^2(x)$ for all $x$.  Define 
     \begin{equation*}
         \begin{split}
             \mathcal{B}_{r}(p)=\sqrt{2\log p}-\frac{1}{\sqrt{2\log p}} \bigg[ &\log\log(p) + \\
             &\log(2\sqrt{\pi})\bigg]+\frac{p}{\sqrt{2\log r}}.
         \end{split}
     \end{equation*}
     Considering the bandwidth condition
     \begin{equation*}
         \label{eq:bw_condition_th1}
          nb_n^9\log n+\frac{(\log n)^3}{nb_n^3}+\Xi_{n}\left\{\frac{b_n^3\log n}{n}+\frac{(\log n)^2}{n^{2} b_n^{\frac{4}{3}}}\right\}\xrightarrow{n\rightarrow\infty}0,
     \end{equation*}
     under the null hypothesis mentioned in \eqref{eq:mu-exist-hypothesis}, for any $z\in\mathbb{R}$,
     \begingroup
     \footnotesize
     \begin{equation*}
     \label{eq:theorem1}
     \begin{split}
         \lim_{n\rightarrow\infty} & P\left[\frac{\sqrt{nb_n}}{\sqrt{\phi(K^*)}}\sup_{x\in \Pi_n}\left\{\left(\frac{\sqrt{\widehat f_X(x)}}{{\sqrt{\widehat\sigma^2(x)}}}\right)\abs{\widehat\mu^*_{\text{diff}}(x)}\right\}\leqslant \mathcal{B}_{m_{n}}(z)\right] =e^{-2e^{-z}},
     \end{split}
     \end{equation*}
     \endgroup
where $\widehat\sigma^2(x)$ is a nonparametric estimate of conditional variance.
\end{theorem}

Following \Cref{lem:theorem3}, the test for a structural break in the conditional mean can be implemented. We first construct the discrete grid of points $\Pi_n$ and evaluate the test statistic
\begin{equation*}
    T(\mu)=\frac{\sqrt{nb_n}}{\sqrt{\phi(K^*)}}\max_{x\in\Pi_n}\left\{\frac{\sqrt{\widehat f_X(x)}}{\sqrt{\widehat\sigma^2(x)}}\abs{\widehat\mu_1^*(x)-\widehat\mu_2^*(x)}\right\}.
\end{equation*}

Then, we reject the null hypothesis of structural stability in conditional mean at level of significance $\eta$ if the observed value of the statistic $T(\mu)$ exceeds the critical value $\mathcal{B}_{m_n}(z_\eta)$, $z_\eta = -\log(-2\log(1-\eta))$ being the $100(1-\eta)\%$ quantile of standard Gumbel distribution.

Along a similar line, the testing problem \eqref{eq:what_to_test_sigma} for the structural stability in the conditional variance function $\sigma^2(\cdot)$ is formulated as
\begin{equation}\label{eq:sigma-exist-hypothesis}
    \begin{split}
        H_0: & \;  \text{ There is no structural break in the conditional variance}, \\
        H_1: & \;  \text{ There is a structural break in the conditional variance}.
    \end{split}
\end{equation}

It can be tested with the statistic
\begin{equation*}
\label{eq:sigmacp}
    \sup_{x\in\mathcal{X}}\left\{\abs{\widehat{\sigma}^{*^2}_{\text{diff}}(x)}\right\}=\sup_{x\in\mathcal{X}}\left\{\abs{\widehat\sigma_1^{*^2}(x)-\widehat\sigma_2^{*^2}(x)}\right\},
\end{equation*}
which, in practice, can be approximated over a dense grid of values of $x$ as defined in $\Pi_n$ before, i.e., we use $\max_{x\in\Pi_n}\left\{\abs{\widehat\sigma_1^{*^2}(x)-\widehat\sigma_1^{*^2}(x)}\right\}$ in practical applications. The below-stated theorem illustrates the behavior of the supremum statistic for the conditional variance  under the null hypothesis of structural stability. 

\begin{theorem}
\label{lem:theorem4}
Along with the previously stated assumptions in \Cref{sec:methods}, let 
\begin{equation*}
    \label{eq:bw_condition_th3}
     nb_n^9\log n+\frac{\log n}{nb_n^4}+\Xi_{n}\left\{\frac{b_n^3\log n}{n}+\frac{(\log n)^2}{n^{2}b_n^{\frac{4}{3}}}\right\}\xrightarrow{n\rightarrow\infty}0.
\end{equation*}
Under the null hypothesis of \eqref{eq:sigma-exist-hypothesis}, for any $z\in\mathbb{R}$,
    \begingroup
     \footnotesize
     \begin{equation*}
    \label{eq:eqth3}
    \begin{split}
        \lim_{n\rightarrow\infty} & P\left[\frac{\sqrt{nb_n}}{\widehat\nu_{\epsilon}\sqrt{\phi(K^*)}}\sup_{x\in\Pi_n}\left\{\sqrt{\frac{\widehat f_X(x)}{\widehat S(x)}}\abs{\widehat\sigma^{*^2}_{\text{diff}}(x)}\right\}\leqslant \mathcal{B}_{m_{n}}(z)\right] \\
        & = e^{-2e^{-z}},
    \end{split}
    \end{equation*}
    \endgroup
    where $\mathcal{B}_{m_n}(z)$ is as defined as \Cref{lem:theorem3}, $\widehat S(x)$ and $\widehat\nu_\epsilon$ are as defined in \Cref{sec:point-wise_asymptotics}.
\end{theorem}

Following \Cref{lem:theorem4}, the test for a structural break in the conditional variance can be implemented similarly as before. Upon constructing the discrete grid of points $\Pi_n$, we evaluate 
\begin{equation*}
    T(\sigma)=\frac{\sqrt{nb_n}}{\widehat\nu_\epsilon\sqrt{\phi(K^*)}}
    \max_{x\in\Pi_n}\left\{\frac{\sqrt{\widehat f_X(x)}}{\sqrt{\widehat S(x)}}\abs{\widehat\sigma_1^{*^2}(x)-\widehat\sigma_2^{*^2}(x)}\right\}.
\end{equation*}
and the decision is to reject the null hypothesis of structural stability in conditional variance at level of significance $\eta$ if the observed value of $T(\sigma)$ is bigger than the critical value $\mathcal{B}_{m_n}(z_\eta)$. 

Note that the key steps to proving the above theorems is to express $\widehat\mu^*_{\text{diff}}(x)$ as a martingale difference sequence and considering its quadratic characteristic matrix. A martingale maximum deviation theorem proposed by \cite{grama2006asymptotic} is then applied to obtain the asymptotic distribution. The technical details of the proof are deferred to Section S2 of the supplementary material. 

Let us now move on to the test of a structural break in both the conditional mean and the conditional variance function. This can be formulated as the following testing problem:
\begin{equation}\label{eq:mu-sigma-exist-hypothesis}
    \begin{split}
        H_0: & \;  \text{ There is no structural break in $\mu(\cdot)$ or in $\sigma^2(\cdot)$}, \\
        H_1: & \;  \text{ There is a structural break either in $\mu(\cdot)$ or in $\sigma^2(\cdot)$} \text{ or in both}.
    \end{split}
\end{equation}

We suggest testing \eqref{eq:mu-sigma-exist-hypothesis} by simultaneously utilizing Theorems \ref{lem:theorem3} and \ref{lem:theorem4}, along with a Holm-Bonferroni correction. In this approach, define $T_{\text{max}}=\max\{\abs{T(\mu)},\abs{T(\sigma)}\}$ and $T_{\text{min}}=\min\{\abs{T(\mu)},\abs{T(\sigma)}\}$. Then, both the null hypotheses in \eqref{eq:mu-exist-hypothesis} and \eqref{eq:sigma-exist-hypothesis} are rejected at $\eta$ level of significance if $T_{\text{min}}\geqslant \mathcal{B}_{m_n}(z_{\frac{\eta}{2}})$ and $T_{\text{max}}\geqslant \mathcal{B}_{m_n}(z_{\eta})$. The performance of this test is illustrated in \Cref{sec:simulation}.

As a last point of discussion in this section, we want to highlight that the above results also facilitate the construction of $100(1-\eta)\%$ simultaneous confidence bands for the differences in the conditional mean function or the same in the conditional variance function in the two halves of the data. These confidence bands are quite effective in obtaining valuable insights about how the covariate process impacts the response variable before and after a potential structural break. We state the expressions for these confidence bands below. The derivations of these bands are straightforward from the previous theorems.

\begin{corollary}
    \label{lem:corollary4}
    Under the conditions stated in \Cref{lem:theorem3}, the confidence band for $\mu_{\text{diff}}(x)$ is 
        \begin{equation*}
\label{eq:cb_mean}
    \widehat{\mu}_{\text{diff}}^*(x) \pm \frac{\sqrt{\phi(K^*) \widehat{\sigma}^2(x)}}{\sqrt{nb_n \widehat{f}_X(x)}} \mathcal{B}_{m_n}(z_\eta);
\end{equation*}
and under the conditions stated in \Cref{lem:theorem4}, the confidence band for $\sigma^{^2}_{\text{diff}}(x)$ is 
\begin{equation*}
\label{eq:sigma-cb}
    \widehat\sigma^{*^2}_{\text{diff}}(x)\pm \frac{\widehat\nu_\epsilon\sqrt{\phi(K^*)\widehat f_X(x)}}{\sqrt{nb_n\widehat S(x)}}\mathcal{B}_{m_n}(z_\eta).
    \end{equation*}
\end{corollary}

\subsection{Multiple structural break detection and localization algorithm}\label{subsec:detection-of-structural-breaks}

While the AMOC framework offers a solid foundation for detecting a single structural break, many real-world processes involve multiple changes at unknown locations. Extending AMOC to handle multiple changepoints is a natural and necessary step to capture such complexity. This is typically achieved by embedding the AMOC test within a sequential segmentation framework, such as the widely used binary segmentation algorithm. Binary segmentation iteratively detects the most significant changepoint in a segment, splits the data at that point, and recursively applies the procedure to the resulting sub-segments. This approach leverages the local detection power of the AMOC test while enabling multiple changepoint detection across the dataset.

The process starts by applying the AMOC test to the full dataset. If no changepoint is detected, the procedure stops, indicating homogeneity. If a changepoint is found, the data is split at the estimated location, and the AMOC test is recursively applied to each sub-segment until no further breaks are detected or segments fall below the user-defined minimum length $L_{\text{min}}$. Since binary segmentation can lead to over-detection or mislocalization, especially with closely spaced changes, a post-detection validation step is used to refine the results and control false positives. Each detected changepoint $b_i$ is re-tested using local sub-data $\{(Y_t, X_t)\}, \; b_{i-1} \leqslant t \leqslant b_{i+1}$ to confirm its significance. Let $\widehat{\mu}_{b_i^-}(\cdot), \widehat{\mu}_{b_i^+}(\cdot)$, and $\widehat{\sigma}^2_{b_i^-}(\cdot), \widehat{\sigma}^2_{b_i^+}(\cdot)$ denote the conditional mean and variance estimates in the two segments, obtained using the kernel estimators from \Cref{sec:point-wise_asymptotics}. The magnitude of the break is assessed by the statistics
\begin{gather}
\label{eq:alternate_teststat}
        \widehat{\mu}_{\text{cp}}(b_i) = \sup_{x \in \mathcal{X}} \left\{ \left|\widehat\mu_{b_i^-}(x) - \widehat{\mu}_{b_i^+}(x)\right| \right\},\quad \widehat{\sigma}^2_{\text{cp}}(b_i) = \sup_{x \in \mathcal{X}} \left\{ \left|\widehat{\sigma}_{b_i^-}^2(x) - \widehat{\sigma}_{b_i^+}^2(x)\right| \right\},
\end{gather}
and we declare $b_i$ to be indeed a structural break if either $\widehat{\mu}_{\text{cp}}(b_i)$ or $\widehat{\sigma}^2_{\text{cp}}(b_i)$ is statistically significantly large. The asymptotic properties of $\widehat{\mu}_{\text{cp}}(b_i)$ and $\widehat{\sigma}^2_{\text{cp}}(b_i)$ follow similarly to those outlined in \Cref{sec:supremum_asymptotics}, and the corresponding tests can be designed in the same fashion. This localization step can help fine-tune the exact position of the changepoint, improving both the precision and the reliability of the detection.

At this stage, it is important to emphasize that an alternative approach is to implement the second stage directly to detect structural break, i.e. to split the data at a random time point $t_0$ and assessing the disparity between the functional characteristics of the data before and after $t_0$ using test statistics $\widehat\mu_{\text{cp}}(t_0)$ or $\widehat\sigma^2_{\text{cp}}(t_0)$ of the form \eqref{eq:alternate_teststat}. Assuming the data contains at most one structural break, the test statistics can then be inverted to detect the location of this break. For practical implementation, we define a sequence of partitions of $\T$ as $\{\T_n\}$, where $\T_n = \{t_1', t_2', \ldots \mid t_j' = 2j b_n, j = 0, 1, 2, \ldots, k_n - 1\}$ with $k_n = \lceil 1/2b_n \rceil$. A structural break in the conditional mean or in the conditional variance can then be estimated as
\begin{equation*}
    \widehat{\tau_0}^{\mu} = \argmax_{t \in \T_n} \left\{ \widehat{\mu}_{\text{cp}}(t) \right\}, \; \; \widehat{\tau_0}^{\sigma} = \argmax_{t \in \T_n} \left\{ \widehat{\sigma}^2_{\text{cp}}(t) \right\}.
\end{equation*}

The consistency of these estimated breaks is guaranteed by the following result.

\begin{theorem}
\label{lem:theorem6}
Let $\tau_0^\mu$ and $\tau_0^\sigma$ denote the true structural breaks in the conditional mean and the conditional variance, respectively. Then, as $n\to\infty$, $\widehat{\tau_0}^{\mu} \xrightarrow{P} \tau_0^{\mu}$ and $\widehat{\tau_0}^{\sigma} \xrightarrow{P} \tau_0^{\sigma}$.
\end{theorem}

The proof follows directly from an application of Berge's maximum theorem; and the details are provided in Section S2 of the supplementary material.

\section{Simulation study}
\label{sec:simulation}

We generate the covariate series $\{X_t\}$ from various data generating processes (DGPs) and obtain the response $\{Y_t\}$ using different conditional mean and variance functions under normal, moderate, and heavy-tailed noise. In the interest of space, the details of the DGPs, noise types, and the conditional structures across segments (determined by randomly introduced structural breaks) are provided in Section S3 of the supplementary material. The conditional mean $\mu(x)$ varies across segments: polynomial regression in segments 1 and 2, log-linear regression in segment 3 (common in econometrics), exponential regression in segment 4 (used in finance and biology), and trigonometric regression in segment 5 (suitable for modeling seasonality). For the variance function $\sigma^2(x)$, we assume homoskedasticity in segment 1, while segments 2–5 feature heteroskedasticity modeled by GARCH, power GARCH, and exponential GARCH structures, following established literature \citep{palm19967}. This setup allows us to assess both the power and localization accuracy of the proposed binary segmentation algorithm.

The nonparametric estimations are done using the parabolic kernel $K(u)=0.75\mathbf{I}_{\left\{\abs{u}\leqslant 1\right\}}$  and the MSE-optimal bandwidth of $b_n=n^{-0.2}$ where $n$ is the length of the data. One may alternatively select the bandwidth using cross-validation procedures, however, our observation has been that the performances of the test with the cross-validated bandwidth are more or less similar to those obtained with the MSE optimal bandwidth. We conduct the simulation study for sample sizes of $n=\left\{500,1000,2000\right\}$. All the simulations performed for evaluating the size and the power of the tests are conducted at 5\% level of significance and under the presence of a single structural break. The conditional mean and variance functions before and after the structural break are taken as mentioned in segment 5 and 2 respectively. To evaluate the performance of the tests, we simulate $1000$ independent samples from a chosen combination of a data generating process (DGP) and noise distribution. In each sample, a single structural break is introduced at a random time point. For each test, the size is measured as the  average empirical Type-I error rate, while the power is assessed by calculating the average proportion of correct rejections of the null hypothesis across all $1000$ samples.
  
\begin{table}[!ht]
\centering
\caption{Performance of the test when there is a single structural break in conditional mean.}
\label{tab:size_power_mu}
    \begin{tabular}{clcccccc}
    \toprule
     & & \multicolumn{2}{c}{Noise: $\gauss(0,1)$} & \multicolumn{2}{c}{Noise: $t_{10}$} & \multicolumn{2}{c}{Noise: Power law} \\
     Sample size & DGP & Size & Power & Size & Power & Size & Power \\
     \midrule
     500 & White Noise & 0.00 & 0.21 & 0.00 & 0.31 & 0.01 & 0.25 \\
     & ARMA-GARCH & 0.01 & 0.23 & 0.00 & 0.29 & 0.00 & 0.26 \\
     & TAR & 0.01 & 0.31 & 0.00 & 0.27 & 0.00 & 0.34 \\
     \midrule
     1000 & White Noise & 0.00 & 0.43 & 0.02 & 0.54 & 0.00 & 0.48\\
     & ARMA-GARCH & 0.00 & 0.37 & 0.00 & 0.51 & 0.00 & 0.48  \\
     & TAR & 0.03 & 0.47 & 0.01 & 0.59 & 0.01 & 0.57 \\ \midrule
     2000 & White Noise & 0.00 &0.53 & 0.01& 0.58 & 0.00 & 0.56\\
     & ARMA-GARCH & 0.04 & 0.55 & 0.02  & 0.58 &0.00  & 0.62\\
     & TAR & 0.01 & 0.58 &0.01 &0.65 &0.00 &0.67 \\
     \bottomrule
    \end{tabular}
\end{table}

For structural breaks in the conditional mean, as illustrated in \Cref{tab:size_power_mu}, the test shows good size control across all settings. The power of the test improves significantly with increasing sample size, starting with moderate power at a sample size of $500$ and reaching up to $67\%$ at a sample size of $2000$, with TAR models consistently demonstrating higher power compared to White Noise and ARMA-GARCH. The performance of the test of structural stability in conditional variance, as well as in both conditional mean and variance is deferred to the supplementary material. We find that when detecting breaks in the conditional variance, the test exhibits similarly good size control, with minimal deviations even at larger sample sizes. The power of the test in case of conditional variance is notably higher than for the corresponding test in mean, achieving relatively strong power even at a sample size of $500$ and further improving with sample size. Overall, the test is more effective at detecting shifts in the conditional variance function than the conditional mean, with the power increasing consistently across all DGPs as the sample size grows, especially for more complex models like TAR. On the other hand, the simultaneous test has around 50\% to 70\% power in smaller sample sizes, working comparatively well in the case of heavy-tailed noises. As the sample size increases the power of the test gets better across all combinations, consistently reaching values in the range of 80\% to 90\%. The performance of the localization is also presented in the supplement and will be helpful for the reader to understand the efficacy of the proposed algorithm.

\section{Application to Bitcoin data}
\label{sec:application}

News sentiment plays a significant role in shaping Bitcoin prices, much like it does with traditional financial assets. Several recent studies have explored public attention to Bitcoin can predict its price behavior \citep{figa2020disentangling}. In this section, we utilize our proposed method to analyze shifts in the relationship between public attention, measured by Google News Interest Score (GNIS) from Google Trends, and Bitcoin's price and volatility. The GNIS data, denoted as $\{G_t\}$, spans January 1 2020 to September 4 2024, assigning a relative score (0-100) to the search volume on ``Bitcoin'', with 100 indicating peak popularity. This score, extracted using the `pytrends' library, reflects search volume as a proportion of all global searches. The daily Bitcoin log-price series covering the same period, denoted as $\{P_t\}$, is sourced from FRED (link: \url{https://fred.stlouisfed.org/series/CBBTCUSD}). The cleaned dataset comprises $1704$ observations, and will be available (along with all R codes) in a GitHub repository maintained by the first author. 

We model the current day Bitcoin log-price as a function of the last day's GNIS, in order to incorporate for a lag in the impact. The model is,
\begin{equation*}
\label{eq:eq_rtmodel}
    P_t = \mu(G_{t-1}) + \sigma(G_{t-1})\epsilon_t,
\end{equation*}
where $\mu(\cdot)$ and $\sigma(\cdot)$ are unknown smooth functions representing the mean and conditional standard deviation of the log-price variable. Note that since we are working with the log-price, $\sigma(\cdot)$ can be thought of as a proxy for the volatility. For the detection procedure, we apply a rule-of-thumb bandwidth, $h_n = n^{-0.2} = 0.2258$. We also ran the same analysis using a cross-validated bandwidth choice, however the detected structural breaks in these two procedures were reasonably close to one another. The algorithm requires a minimum of $200$ days between consecutive structural breaks. We detect two structural breaks in the mean price level, no breaks were found in the conditional variance of the price and in the return series. The detected breaks in the conditional mean of the Bitcoin price are illustrated in \Cref{fig:sbplot}. 

\begin{figure}[!ht]
\centering
\captionsetup{justification=centering}
\centerline{\includegraphics[width=0.8\columnwidth,keepaspectratio]{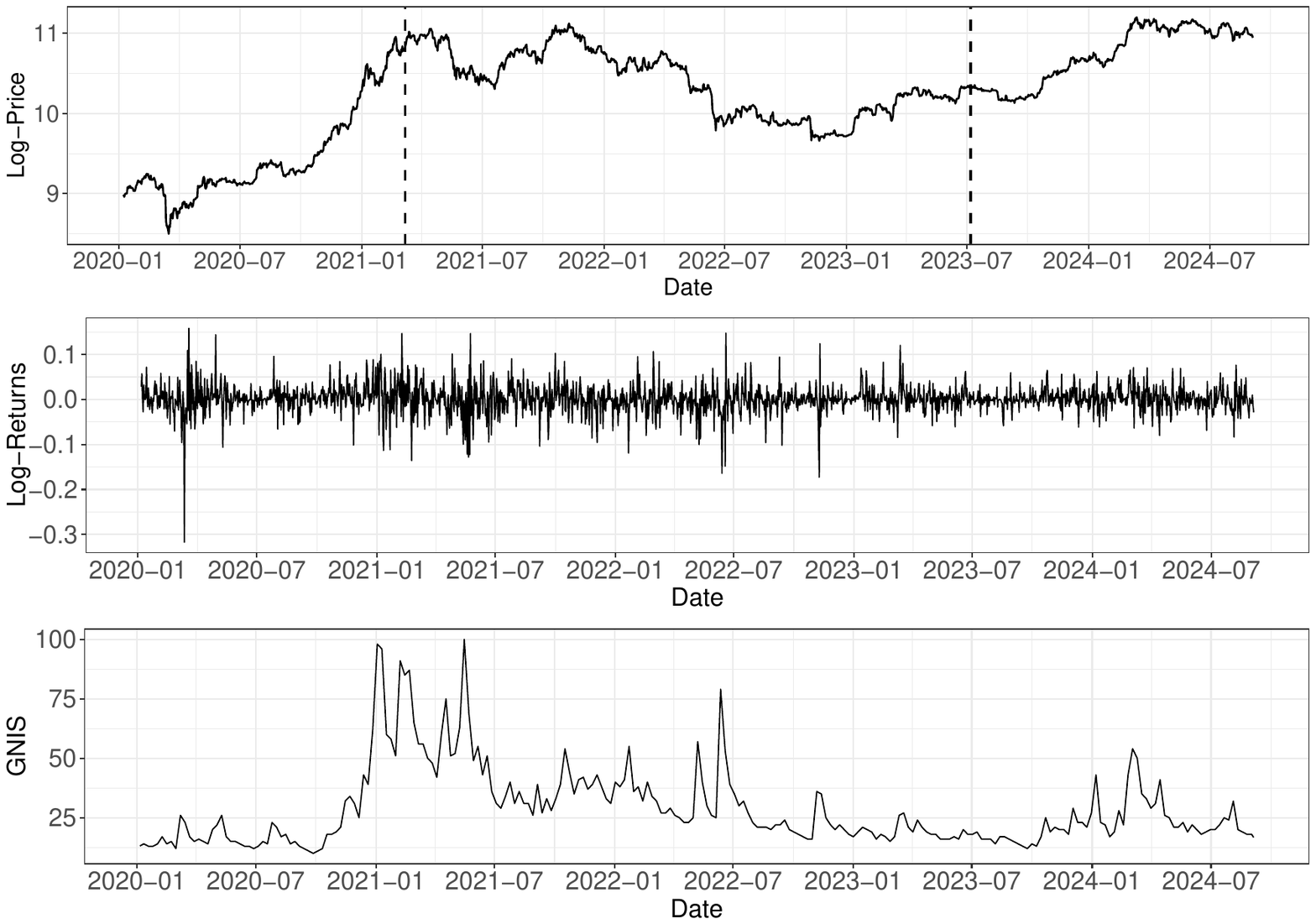}}
\caption{(Top) Bitcoin price (log-transformed) and the detected structural breaks in the impact of public attention on its average level; (middle) log-return of Bitcoin price; (bottom) GNIS series for the time period.}
\label{fig:sbplot}
\end{figure}

The first break, detected on 7 March 2021, aligns with Bitcoin’s rapid rise to an all-time high in early 2021, when institutional interest in Bitcoin surged and the cryptocurrency gained widespread media attention. This period saw companies like Tesla investing in Bitcoin, which drove significant news coverage and heightened interest in the asset, contributing to the price surge. The second structural break on 8 July 2023, suggests another structural shift, possibly linked to market stabilization after major corrections. Prior to this date, the market was volatile, and news was largely negative (declining trend), reflecting uncertainty and regulatory pressures. However, by July, the tone of news coverage shifted to more positive stories, focusing on institutional adoption, technological advancements, and market stabilization. This change in sentiment likely contributed to Bitcoin’s price recovery, with the news acting as a stabilizing factor, reinforcing the upward momentum in the market. For better understanding of the market behavior before and after the structural breaks, the descriptive statistics of the price, returns and GNIS for the entire dataset as well as the three segments generated by the two structural breaks are presented in Section S4 in the supplement.

\section{Concluding remarks}
\label{sec:conclusion}

In this paper, we present a test designed to detect structural breaks occurring within the middle of a dataset. The theoretical foundations of the test's asymptotic properties are rigorously established. Additionally, the test can be adapted to identify structural breaks in the conditional mean, conditional variance, or both, at unknown points in time. The effectiveness of the proposed test is demonstrated through a comprehensive simulation study, which covers a wide array of linear and nonlinear data-generating processes, incorporating various contamination levels (thin-tailed, moderately heavy-tailed, and heavy-tailed distributions). 

Even though in the current work we only focus on detecting structural breaks in the conditional mean and variance only, the same theoretical idea can be utilized to extend the method to detect breaks in higher order conditional moment functions as well. We leave the detailed derivations to a future work. It is also possible to derive similar results using the local linear or spline estimator instead of a kernel estimator with slight modifications in the bandwidth conditions. That can be another interesting future extension of the work. Additionally, our current study focuses on univariate response series, but one may attempt to extend the proposed methodology to high-dimensional setting. This approach would enable the modeling of structural breaks as part of a broader stochastic process, potentially offering richer insights into the underlying dynamics of financial data.

\section*{Data availability statement}
The cleaned data, codes and outputs will be provided in a GitHub repository maintained by the first author.

\bibliography{references}


\newpage

\begin{center}
	{\Large {\bf Supplementary materials}}
\end{center}

\setcounter{section}{0}
\setcounter{equation}{0}
\setcounter{figure}{0}
\setcounter{table}{0}
\def\theequation{S\arabic{section}.\arabic{equation}}
\def\thesection{S\arabic{section}}
\def\thetable{S\arabic{table}.\arabic{table}}
\def\thefigure{S\arabic{figure}.\arabic{figure}}

\section{A note on the choice of the bandwidth}
\label{sec:bw-selection}

The asymptotic theory derived in the main manuscript emphasize the critical role of selecting an appropriate bandwidth for the estimation of the conditional mean and variance functions in order to get desired performance in the proposed tests. The conditions set forth in these results impose specific constraints on both the bandwidth $b_n$ and the dependence range of the covariate $\{X_t\}$. For instance, the first part of the bandwidth condition in Theorem 1, $nb_n^9\log n \to 0$, ensures that the bandwidth is not too large. Meanwhile, the second part, $(\log n)^3/nb_n^3 \to 0$, prevents the bandwidth from being too small. The third condition involving $\Xi_{n}$ guarantees that the dependence range of the covariate process $\{X_t\}$ remains within an appropriate range. Notably, for short-range dependent (SRD) processes, $\Xi_{n} = O(n)$, ensuring that the bandwidth condition is met as long as the first two terms are $o(1)$. This aligns with choosing the bandwidth as
\begin{equation*}
    b_n = O(n^{-\beta}), \quad \beta \in \left(\frac{1}{9}, \frac{1}{3}\right).
\end{equation*}

The situation for long-range dependent (LRD) processes is more complex. Consider, for example, a zero-mean \textit{iid} process $\{\eta_t\}$ with $\eta_0\in\mathcal{L}^q, q\geqslant 2$, and define a process of the form
\begin{equation*}
\label{eq:example_process}
    X_t = \sum_{j=0}^{\infty} a_j\eta_{t-j}, \quad \text{with } a_j = \frac{l(j)}{j^\kappa}, \quad \kappa \in \left(\frac{1}{2},1\right],
\end{equation*}
where $l(\cdot)$ is a slowly varying function. The structure defined above encompasses a broad range of linear and nonlinear processes. For example, if we set 
\begin{equation*}
    a_j = \frac{\Gamma(j+d)}{\Gamma(j+1)\Gamma(d)},
\end{equation*}
where $\Gamma$ denotes the incomplete gamma function, with $d \in \left(0, 0.5\right)$, we obtain the long-range dependent FARIMA process with orders $(0,d,0)$. This specification can also produce various nonlinear processes, including the class of nonlinear AR and ARCH models.

Our setting also accommodates heavy-tailed $\{\eta_t\}$. It can be shown that the process $\{X_t\}$ of this form exhibits LRD characteristics. Elementary calculations yield
\begin{equation*}
    \Xi_{n} = \begin{cases}
        O\left(n^{3-2\kappa}l^2(n)\right) & \text{if } \kappa \in \left(\frac{1}{2},1\right), \\
        O\left(n\left(\sum_{i=1}^{n} \frac{\abs{l(i)}}{i}\right)^2\right) & \text{if } \kappa = 1.
    \end{cases}
\end{equation*}

Now, to obtain a bandwidth condition similar to those in Theorem 1 and Theorem 2, we require $\kappa \in \left(\frac{17}{26}, 1\right]$, under which the bandwidth $b_n=O(n^{-\beta})$ should satisfy
\begin{equation*}
    \beta \in \left(\max\left\{\frac{1}{9}, \frac{2-2\kappa}{3}\right\}, \min\left\{\frac{1}{3}, \frac{3(2\kappa-1)}{4}\right\}\right).
\end{equation*}

In particular, the MSE-optimal bandwidth $n^{-0.2}$ nearly satisfies both the requirements of LRD and SRD processes. In parallel, a common approach to obtaining the optimal bandwidth in nonparametric regression is through a cross-validation procedure, which involves selecting the bandwidth that minimizes a specified loss criterion. One such criterion is the mean squared error (MSE) criterion \citep{hall1991optimal}, defined as
\[
L_{\text{MSE}} = E\left[\left(\widehat{g}(x \mid b_n) - g(x)\right)^2\right], \quad g(\cdot) \equiv \mu(\cdot) \text{ or } \sigma(\cdot),
\]
where \(\widehat{g}(x \mid b_n)\) is the estimate of \(g(x)\) using the Nadaraya-Watson kernel estimator with bandwidth \(b_n\). More generally, the bandwidth can also be chosen by minimizing the conditional weighted mean squared error
\[
L_{\text{WMSE}} = \int_{-\infty}^{\infty} E\left[\left(\widehat{g}(x \mid b_n) - g(x)\right)^2\right] \omega(x) \, dx,
\]
where \(\omega(\cdot)\) is a suitable weight function with compact support. Another pertinent work was done by \cite{giordano2008neural} who designed a feed forward neural network with one hidden layer that is trained to minimize the prediction error of the local linear regression estimates. While this approach may be an effective alternative, it is important to note that the authors worked with only a nonlinear autoregressive model structure.

\section{Proofs}\label{sec:proofs}

We start by proving a few prerequisite lemmas. Recall that the filtration generated by $\{X_t\}$ is $\{\mathcal{F}_t\}$ and the joint filtration generated by $\{X_t\}$ and $\{\epsilon_t\}$ is $\{\mathcal{G}_t\}$. Here, $\{X_t\}$ is independent of $\{\epsilon_t\}$. Denote by $\mathcal{K}$ the set of kernel functions which are symmetric, bounded, are defined on the bounded support $[-1,1]$ and have bounded derivative. Hereafter, we use the notation $\sum_{\mathcal{R}}$ to indicate that the intended result for the summation term holds for both $\mathcal{R}=\tplus$ and $\mathcal{R}=\tminus$. For notational convenience throughout this section, we shall use $K_{b_n}(u)$ to denote the term $K(b_n^{-1}u)$. 

\begin{lemma}
\label{lem:lemma1}
    Define
\begin{equation*}
\mathcal{I}_{n}(x)=\sum_{\mathcal{R}}{\left\{f_X(x\mid \mathcal{F}_{t-1})-E[f_X(x\mid \mathcal{F}_{t-1})]\right\}}.
\end{equation*}

Recall the definition of $\Xi_n$ from Assumption 2 in the main paper. Then, for a fixed $x\in\mathbb{R}$ and $\Delta>0$
\begin{equation*}
\norm{\sup_{\abs{x}\leqslant \Delta}\left\{\abs{\mathcal{I}_{n}(x)}\right\}}_2=O\left(\sqrt{\Xi_{n}}\right).
\end{equation*}
\end{lemma}
\begin{proof}
\sloppy
    The proof follows from Theorem 1 of \cite{wu2007strong}.
\end{proof}

\begin{lemma}
\label{lem:lemma2}
    Let $f_1$ and $f_2$ be measurable functions such that $f_2(\epsilon_0)\in\mathcal{L}^2$ and each of $f_1(x)$ and $f_X(x)$ belongs to $\mathcal{C}^{0}\left(x-\delta, x+\delta\right)$ for some $\delta>0$. Further, assume that $f_1$ and $f_X$ do not vanish anywhere on $\mathcal{X}$.  Define, for any $K\in\mathcal{K}$,
    \begin{equation*}
\upsilon_t(x)=\frac{f_1(X_t)\left[f_2(\epsilon_t)-E(f_2(\epsilon_t))\right]K_{b_n}(x-X_t)}{f_1(x)\sqrt{nb_nV(f_2(\epsilon_0))\phi(K)f_X(x)}}.
    \end{equation*}

Then, assuming $b_n\rightarrow 0, nb_n\rightarrow\infty$ and $n^{-2}\Xi_{n}\rightarrow 0$ as $n\rightarrow\infty$, we have for a fixed $x\in\mathcal{X}$,
 \begin{equation*}
 \label{eq:lemma_1_statement}
     S_n(x)=\sum_{\mathcal{R}}{\upsilon_t(x)}\xrightarrow{d}\gauss(0,1).
 \end{equation*}
 \end{lemma}

\begin{proof}
\sloppy
Due to the independence of $\{X_t\}$ and $\{\epsilon_t\}$, it is straightforward to show that $\{\upsilon_t\}$ forms a martingale difference sequence with respect to the filtration $\{\mathcal{G}_t\}$. Now, define
\begin{gather*}
\gamma_t(x)=f_1^2(X_t)K_{b_n}(x-X_t).
\end{gather*}

Setting $U_t= \gamma_t(x)-E(\gamma_t(x)\mid \mathcal{F}_{t-1})$ and $V_t=E(\gamma_t(x)\mid \mathcal{F}_{t-1})-E(\gamma_t)$, we can write
\begin{equation*}
    \sum_{\mathcal{R}}{\left\{\gamma_t-E(\gamma_t)\right\}}= \sum_{\mathcal{R}}{U_t}+\sum_{\mathcal{R}}{V_t}.
\end{equation*}

Since $\{U_t\}$ forms a martingale difference sequence with respect to the filtration $\{\mathcal{F}_t\}$ and $E(U_t^2)=O(b_n)$, we can write $\sum_{\mathcal{R}}{U_t}=O_p(\sqrt{nb_n})$. On the other hand, taking advantage of the fact that $f_1(x)\in\mathcal{C}^{0}(x-\delta,x+\delta)$ and $K\in\mathcal{K}$, \Cref{lem:lemma1} implies that
\begin{equation*}
\label{eq:eq_rn}
    \norm{\sum_{\mathcal{R}}{V_t}}_2=O\left(b_n\sqrt{\Xi_n}\right).
\end{equation*}

Simple calculations show that since $b_n\rightarrow 0$, $nb_n\rightarrow\infty$ and $n^{-2}\Xi_n\rightarrow 0$,
\begin{equation}
\label{eq:martingale_2}
    \sum_{\mathcal{R}}{E\left(\upsilon_t^2\mid \mathcal{G}_{t-1}\right)}=\frac{\sum_{\mathcal{R}}{U_t}+\sum_{\mathcal{R}}{V_t}+\sum_{\mathcal{R}}{E(\gamma_t)}}{nb_n\phi(K)f_X(x)f_1^2(x)}\xrightarrow{P}1.
\end{equation}

Next, since $f_1(x)\in\mathcal{C}^0(x-\delta,x+\delta)$ and $K$ is bounded, we must have, for some $c>0$ and large enough $n$,
\begin{equation*}
    \sup_{u\in\mathbb{R}}\left\{\abs{f_1(u)K_{b_n}(x-u)}\right\}\leqslant c.
\end{equation*}

For any $s>0$, define $d(x)=c^{-1}sf_1(x)\sqrt{n b_n V(f_2(\epsilon_0))\phi(K)f_X(x)}$. Again, exploiting the independence of $\{X_t\}$ and $\{\epsilon_t\}$, we deduce
\begin{equation}
\label{eq:martingale_1}
        \sum_{\mathcal{R}}{E\left(\upsilon_t^2(x)\mathbf{1}_{\left\{\abs{\upsilon_t(x)}\geqslant s\right\}}\right)} \leqslant \frac{E\left(\gamma(X_t)\right)  E\left(\left[f_2(\epsilon_0)-E(f_2(\epsilon_0))\right]^2\mathbf{1}_{\left\{\abs{f_2(\epsilon_0)-E(f_2(\epsilon_0))}\geqslant d(x)\right\}}\right)}{b_nV(f_2(\epsilon_0))f_1^2(x)\phi(K)f_X(x)}
    \rightarrow 0.
\end{equation}
 
Finally, combining \eqref{eq:martingale_2} and \eqref{eq:martingale_1}, we can conclude that the sequence $\{\upsilon_t(x)\}$ satisfies the conditions for martingale central limit theorem, which implies the statement of the lemma.
\end{proof}

\begin{lemma}
\label{lem:lemma3}
Let $K\in\mathcal{K}$ and $f_2(\epsilon_0)\in\mathcal{L}^3$. Assume that $f_X$ and $f_1$ never vanish on $\mathcal{X}$, and each of $f_X$ and $f_1$ belongs to $\mathcal{C}^4(\mathcal{X}(\delta))$ for some $\delta>0$. Choose the bandwidth $b_n$ such that
    \begin{equation}
        \label{eq:bw_condition_lemma3}
         b_n^{\frac{4}{3}}\log n+\frac{(\log n)^3}{nb_n^3}+\frac{\Xi_{n}(\log n)^2}{n^{2}b_n^{\frac{4}{3}}}\xrightarrow{n\rightarrow\infty} 0.
    \end{equation}

Recall $S_n(x)$ from \Cref{lem:lemma2}. Then
    \begin{equation*}
        \label{eq:lemma3_statement}
        \lim_{n\rightarrow\infty}P\left(\sup_{x\in\Pi_n}\left\{\abs{S_n(x)}\right\}\leqslant \mathcal{B}_{m_n}(z)\right)=e^{-2e^{-z}},
    \end{equation*}
where $\{\Pi_n\}$ and $\mathcal{B}_n(\cdot)$ are as defined in Section 3.2 of the main manuscript.
\end{lemma}

\begin{proof}
    \sloppy
    For a fixed $x\in\mathcal{X}$, recall the definition of $\{\upsilon_t(x)\}$ from \Cref{lem:lemma2}. Choose $k$ random points from $\Pi_n$ and denote them by $x_{t_{j_1}},x_{t_{j_2}},...,x_{t_{j_k}}$. Set $\mathbf{x}=(x_{t_{j_1}},x_{t_{j_2}},...,x_{t_{j_k}})'$, and for each $l=1,2,\ldots,k$  define $S_{n}(x_{t_{j_l}})=\sum_{\mathcal{R}}{\upsilon_{t}(x_{j_l})}$  and $\mathbf{S}_{n,k}(\mathbf{x})=(S_{n}(x_{t_{j_1}}),S_{n}(x_{t_{j_2}}),...,S_{n}(x_{t_{j_k}}))'$. Without loss of generality, we may assume that the first $p$ elements of $\mathbf{x}$ are from $\mathcal{X}_-$ and rest are from $\mathcal{X}_+$, where $\mathcal{X}_-$ and $\mathcal{X}_+$ respectively denote the ranges of the covariate $X$ in the segments $\tminus$ and $\tplus$.  Let $\mathcal{Q}$ be the quadratic characteristic matrix of $\mathbf{S}_{n,k}$, i.e., for any $\mathbf{z} \in \R^k$, 
        \begin{equation*}
            \mathcal{Q}(\mathbf{z})=\sum_{\mathcal{R}}{E\left(\upsilon_{t}(\mathbf{z})\upsilon_{t}^\intercal(\mathbf{z})\big|\mathcal{G}_{t-1}\right)}.
        \end{equation*}
     
    Since the support of the kernel $K$ is $[-1,1]$, it is clear that $\mathcal{Q}$ is a diagonal matrix. The diagonal terms of $\mathcal{Q}$ are given by
    \begin{gather*}
        \mathcal{Q}_{r,r}=
        \frac{1}{f_1^{2}(x_{t_{j_r}})n^{\alpha}b_n\phi(K)f_X(x_{t_{j_r}})}\sum_{\mathcal{R}}{E\left[\gamma_{t}(x_{t_{j_r}})\mid \mathcal{F}_{t-1}\right]},
        \label{eq:fut_qrrform}
    \end{gather*}
   for each $r=1,2,\ldots,k$. Now, define,
    \begin{gather*}
    \vartheta_t(r) = f_1^2(X_t)K^2_{b_n}\left(x_{t_{j_r}}-X_t\right) - E\left[f_1^2(X_t)K^2_{b_n}\left(x_{t_{j_r}}-X_t\right)\mid \mathcal{F}_{t-1}\right],\\
        \varrho_t(r)= E\left[f_1^2(X_t)K^2_{b_n}\left(x_{t_{j_r}}-X_t\right) \mid \mathcal{F}_{t-1}\right]-E\left[f_1^2(X_t)K^2_{b_n}\left(x_{t_{j_r}}-X_t\right)\right].
    \end{gather*}
    
    Using an argument similar to that used in \Cref{lem:lemma2}, it is straightforward to show that 
    \begin{equation}
    \label{eq:varrho_vartheta_norm}
        \norm{\sum_{\mathcal{R}}{\vartheta_t(r)}}_2=O\left(\sqrt{nb_n}\right)\text{ and }\norm{\sum_{\mathcal{R}}{\varrho_t(r)}}_2=O\left(b_n\sqrt{\Xi_{n}}\right).
    \end{equation}
    
    On the other hand, using Taylor's expansion, we can write
    \begin{equation}
        \label{eq:taylor_expansion}
        \abs{\sum_{\mathcal{R}}{E\left[f_1^2(X_t)K^2_{b_n}\left(x_{t_{j_r}}-X_t\right)\right]}-nb_n\mathcal{Q}_{r,r}}=O\left(nb_n^3\right).
    \end{equation}
    
    Combining \eqref{eq:varrho_vartheta_norm} and \eqref{eq:taylor_expansion}, we have,
    \begin{gather}
        \label{eq:qrr_norm_2}
        E\left(\abs{\mathcal{Q}_{r,r'}-\mathcal{I}_{r,r'}}^{\frac{3}{2}}\right)=O\left(\left(\frac{1}{\sqrt{nb_n}}+b_n^4+\frac{\sqrt{\Xi_{n}}}{n}\right)^{\frac{3}{2}}\right)\forall \; r,r',
    \end{gather}
    where $\mathcal{I}$ is the $k\times k$ identity matrix. Elementary calculations also show that
    \begin{equation}
\label{eq:zeta_3_sum}
\sum_{\mathcal{R}}{\abs{\upsilon_t(x_{t_{j_r}})}^3}=O\left(\frac{1}{\sqrt{nb_n}}\right),
    \end{equation}
   so that \eqref{eq:qrr_norm_2} and \eqref{eq:zeta_3_sum} together imply 
    \begin{equation*}
    \label{eq:ghcondition1}
        E\left(\abs{\mathcal{Q}_{r,r'}-\mathcal{I}_{r,r'}}^{\frac{3}{2}}\right)+\sum_{\mathcal{R}}{\abs{\upsilon_t(x_{t_{j_r}})}^3}=\frac{1}{\sqrt{nb_n}}+b_n^3+\frac{\Xi_{n}^{\frac{3}{4}}}{n^{\frac{3}{2}}}.
    \end{equation*}
    
    Next, let $\mathcal{A}_{j_r}$ be the event $\left\{\abs{S_n(x_{t_{j_r}})}\geqslant \mathcal{B}_{m_n}(z)\right\}$ and $\mathcal{E}_{m_n}=\bigcup_{j=0}^{m_n}{\mathcal{A}_{j_r}}$. Under \eqref{eq:bw_condition_lemma3}, it can be easily verified that, for any fixed $x\in\mathcal{X}$,
    \begin{equation*}
    \label{eq:bmnz_condition}
        \left(\frac{1}{\sqrt{nb_n}}+b_n^3+\frac{\Xi_{n}^{\frac{3}{4}}}{n^{\frac{3}{2}}}\right)\left\{1+\mathcal{B}_{m_n}(z)\right\}^4e^{\frac{\mathcal{B}_{m_n}^2(z)}{2}}\xrightarrow{n\rightarrow\infty}0.
    \end{equation*}
   
   Now, using Theorem 1 of \cite{grama2006asymptotic}, we can write
    \begin{equation*}
        \label{eq:lemma3_intersection}
        P\left(\bigcap_{r=1}^{k}{\mathcal{A}_{j_r}}\right)=\left(\frac{2e^{-z}}{m_n}\right)^k\left(1+o(1)\right).
    \end{equation*}
   
   The lemma hence follows using the principle of inclusion and exclusion.
\end{proof}

\begin{proof}[Proof of Theorem 1]
\sloppy
    Define 
\begin{equation*}
    \widehat w_{n}^{(1)}(x)= \frac{f_{X}(x)}{\widehat f_{X}^{(1)}(x)}\quad\text{and}\quad\widehat w_{n}^{(2)}(x)= \frac{f_{X}(x)}{\widehat f_{X}^{(2)}(x)}.
\end{equation*}
\Cref{lem:lemma1} can be utilized to prove that
    \begin{equation*}
    \widehat w_{n}^{(i)}(x)=1+O_{p}\left(\frac{1}{\sqrt{{nb_n}}}+b_n^2+\frac{\Xi_{n}^{\frac{1}{2}}}{n}\right),i=1,2.
\end{equation*}

Further, defining
\begin{equation*}
\begin{split}
    \mathcal{U}_{n}^{(1)}(x) &= \frac{1}{nb_nf_{X}(x)}\sum_{\tminus}{(\mu_{1}(X_{t})-\mu_{1}(x))K_{b_n}(x-X_t)}, \\
    \mathcal{U}_{n}^{(2)}(x) &= \frac{1}{nb_nf_{X}(x)}\sum_{\tplus}{(\mu_{2}(X_{t})-\mu_{2}(x))K_{b_n}(x-X_t)},
\end{split}
\end{equation*}
it can be shown that, for $i=1,2$,
\begin{equation*}
\label{eq:un1asym}
\mathcal{U}_{n}^{(i)}(x)=b_n^{2}\psi(K)\rho_{\mu_{1}}(x)+O_{p}\left(\sqrt{\frac{b_n}{n}}+b_n^{2}+\frac{\sqrt{\Xi_{n}}b_n}{n}\right).
\end{equation*}

Therefore, under bandwidth conditions similar to those specified in \Cref{lem:lemma2}, we have
\begin{equation*}
    \label{eq:mu_uv}
    \left(\widehat \mu_{1}(x)-\widehat \mu_{2}(x)\right)-\left(\mu_{1}(x)-\mu_{2}(x)\right)-\left(b_n^{2}\psi(K)\rho_{\mu_{1}}(x)-b_n^{2}\psi(K)\rho_{\mu_{2}}(x)\right)= \mathcal{V}_{\T}(x)
\end{equation*}
where 
\begin{equation*}
    \mathcal{V}_{\T}(x)= \frac{1}{nb_nf_{X}(x)}\sum_{\tminus}{\sigma(X_{t})\epsilon_{t}K_{b_n}(x-X_t)}-\frac{1}{nb_nf_{X}(x)}\sum_{\tplus}{\sigma(X_{t})\epsilon_{t}K_{b_n}(x-X_t)}.
\end{equation*}

Theorem 1 can then be proved using \Cref{lem:lemma2} with specific forms of the functions $f_1$ and $f_2$. The proofs of Corollary 1 and Corollary 2 follow along similar lines.
\end{proof}

\begin{proof}[Proof of Theorem 2]
Define 
\begin{equation*}
\begin{split}
    W_{n}^{(1)}(x) &= \frac{1}{nb_nf_{X}(x)}\sum_{\tminus}{\sigma_{1}(X_{t})\epsilon_{t}K_{b_n}(x-X_t)}, \\
    W_{n}^{(2)}(x) &= \frac{1}{nb_nf_{X}(x)}\sum_{\tplus}{\sigma_{2}(X_{t})\epsilon_{t}K_{b_n}(x-X_t)}.
\end{split}
\end{equation*}
One can similarly define $W_n^{*^{(1)}}(x)$ and $W_n^{*^{(2)}}(x)$ by replacing $K$ by $K^*$ in the above equations. Further, define
\begin{equation*}
        r_n =\sqrt{\frac{b_n\log n}{n}}+b_n^4+\frac{\sqrt{b_n\Xi_{n}}}{n},        q_n =\sqrt{\frac{\log n}{nb_n}}+b_n^2+\frac{\sqrt{\Xi_{n}}}{n},
        \chi_n = \sqrt{\frac{\log n}{nb_n}}+\frac{\log n}{\sqrt{n^{3}b_n^5}},
\end{equation*}
and denote $\Delta_n =r_n+q_n\left(b_n^2+\chi_n\right)$. Simple calculations show that
\begin{equation*}
    \widehat \mu^*(x)-\mu(x)= 
\begin{cases}
    W_{n}^{*(1)}(x)+O_{p}(\Delta_{n}) & \text{in } \tminus,\\
     W_{n}^{*(2)}(x)+O_{p}(\Delta_{n})  & \text{in } \tplus. 
\end{cases}
\label{eq:eqsigma2}
\end{equation*}

Assuming $\Delta_{n}=o\left(\left(n b_{n}\right)^{-\frac{1}{2}}\right)$, we have
\begin{equation*}
\label{eq:s1_convergence}
    \Delta_{n}\sum_{\tminus}{\left(Y_{t}-\widehat \mu^{*}(X_{t})\right)K_{b_n}(x-X_t)}\xrightarrow{P}0.
\end{equation*}

Then, we can write
\begin{equation}
\label{eq:sigma1_expression1}
    \widehat \sigma_{1}^{2}(x)= \frac{1}{nb_n\widehat f_X(x)}\sum_{\tminus}{\left(\sigma_{1}(X_{t})\epsilon_{t}-W_{n}^{*(1)}(X_{t})+O_{p}(\Delta_{n})\right)^{2}K_{b_n}(x-X_t)}.
\end{equation}

Since $\sup_{x\in\tminus}\left\{\abs{W_n^{*^{(1)}}}(x)\right\}=O_p(\chi_n)$, it can be shown that 
\begin{equation}
\label{eq:sigma1_expression2}
    \begin{split}
    &\sum_{\tminus}{\left(\sigma_{1}(X_{t})\epsilon_{t}-W_{n}^{*(1)}(X_{t})+O_{p}(\Delta_{n})\right)^{2}K_{b_n}(x-X_t)} = \sum_{\tminus}{\left(\sigma_{1}(X_{t})\epsilon_{t}\right)^{2}K_{b_n}(x-X_t)} \\
    &\qquad + O_{p}(\Delta_{n})\sum_{\tminus}{\sigma_{1}(X_{t})\epsilon_{t}K_{b_n}(x-X_t)} -2\sum_{\tminus}{\sigma_{1}(X_{t})\epsilon_{t}W_{n}^{*(1)}(X_{t})K_{b_n}(x-X_t)}.
    \end{split}
\end{equation}

Combining \eqref{eq:sigma1_expression1} and \eqref{eq:sigma1_expression2} and taking absolute value on both sides, we can write
\begin{equation*}
\label{eq:sigma1_expression4}
    \widehat \sigma_{1}^2(x)\leqslant \frac{T_{\tminus}(x)}{nb_n\widehat f_X(x)}+\frac{2\abs{L_{\tminus}(x)}+O_{p}(\Delta_{n})J_{\tminus}(x)}{nb_n\widehat f_X(x)}.
\end{equation*}
where
\begin{align*}
      T_{\tminus}(x)&=\sum_{\tminus}{(\sigma_{1}(X_{t})\epsilon_{t})^{2}K_{b_n}(x-X_t)}, \\ 
    L_{\tminus}(x)&=\sum_{\tminus}{\sigma_{1}(X_{t})\epsilon_{t}W_{n}^{*(1)}(X_{t})K_{b_n}(x-X_t)},\\
     J_{\tminus}(x)& = \sum_{\tminus}{\sigma_{1}(X_{t})\abs{\epsilon_{t}}K_{b_n}(x-X_t)}.
\end{align*} 

Since $J_{\tminus}(x)= O_{p}\left(nb_n\left(1+q_n+\chi_n\right)\right)$ and $\sup_{x\in\tminus}\left\{\abs{L_{\tminus}(x)}\right\}= O_{p}\left(b_n^{-\frac{3}{2}}\right)$, under pre-specified bandwidth conditions we can write, $ \widehat\sigma_{1}^{2}(x)=(nb_n\widehat f_X(x))^{-1}T_{\tminus}(x)$. Defining 
\begin{align*}
    D_{\tminus}(x) = \sum_{\tminus}{\left(\sigma_{1}^{2}(X_{t})-\sigma_{1}^{2}(x)\right)K_{b_n}(x-X_t)}, \;
    E_{\tminus}(x) = \sum_{\tminus}{\sigma_{1}^{2}(X_{t})\left(\epsilon_{t}^{2}-1\right)K_{b_n}(x-X_t)},
\end{align*} 
we can write
\begin{equation}
    \widehat\sigma_{1}^{2}(x)-\sigma_{1}^{2}(x)= \frac{D_{\tminus}(x)+E_{\tminus}(x)}{nb_n\widehat f_X(x)}.
    \label{eq:sigma1_expresion6}
\end{equation}

On the other hand,
\begin{equation}
\label{eq:dminus_convergence}
     \sup_{x\in\tminus}\left\{\abs{\frac{D_{\tminus}(x)}{nb_n\widehat f_X(x)}-b_n^2\psi(K)\rho_{\sigma_{1}}(x)}\right\}= O_{p}(r_n),
\end{equation}
where $\rho_{\sigma_i}(x)=2(\sigma_i'(x))^2+2\sigma_i(x)\sigma_i''(x)+4(f_X(x))^{-1}\sigma_i(x)\sigma_i'(x)f_X'(x)\text{ for } i=1,2$. Now, using \eqref{eq:dminus_convergence} in \eqref{eq:sigma1_expresion6}, we can write
\begin{equation*}
     \widehat \sigma_{1}^{2}(x)-\sigma_{1}^{2}(x)-b_n^2\psi(K)\rho_{\sigma_{1}}(x)= \frac{E_{\tminus}(x)}{nb_n\widehat f_X(x)}.
    \label{eq:eqsigma17}
\end{equation*}

Performing a similar calculation in $\tplus$ and implementing the jackknife correction using the kernel $K^*$, we have
\begin{equation*}
    \left(\widehat \sigma_{1}^{*^2}(x)-\widehat\sigma_{2}^{*^2}(x)\right)-\left(\sigma_{1}^{2}(x)-\sigma_{2}^{2}(x)\right)=\frac{E_{\tminus}^*(x)}{nb_n\widehat f_X(x)}-\frac{E_{\tplus}^*(x)}{nb_n\widehat f_X(x)}.
    \label{eq:eqsigma21}
\end{equation*}

Theorem 2 now follows from \Cref{lem:lemma2} by taking the appropriate forms of $f_1$ and $f_2$.
\end{proof}

\begin{proof}[Proofs of Theorems 3 and 4]
Fix $n$, pick any $k$ real numbers from $\Pi_n\subset \mathcal{X}$ and denote them by $x_{t_{j_1}},x_{t_{j_2}},...,x_{t_{j_k}}$. Let $\mathbf{x}=(x_{t_{j_1}},x_{t_{j_2}},...,x_{t_{j_k}})'$. Under the assumption that the conditional variance process remains the same throughout the time horizon, following an argument similar to that used in the proof of Theorem 1, the null hypothesis indicates that for any $x\in\mathcal{X}$,
\begin{equation*}
\begin{split}
    \widehat\mu_1^*(x)-\widehat\mu_2^*(x) &= \frac{1}{nb_nf_X(x)}\sum_{\tminus}{\sigma(X_t)\epsilon_t K^*_{b_n}(x-X_t)} 
- \frac{1}{nb_nf_X(x)}\sum_{\tplus}{\sigma(X_t)\epsilon_tK^*_{b_n}(x-X_t)}.
\end{split}
\end{equation*}

Define
\[
    \zeta_{t}^{(1)}(x)= 
    \frac{\sigma(X_t)\epsilon_{t}K^*_{b_n}(x-X_t)}{\sigma(x)\sqrt{nb_n\phi(K^*)f_{X}(x)}}\mathbf{I}_{\{t\in\tminus\}}, \quad 
    \zeta_{t}^{(2)}(x)= \frac{\sigma(X_t)\epsilon_{t}K^*_{b_n}(x-X_t)}{\sigma(x)\sqrt{nb_n\phi(K^*)f_{X}(x)}}\mathbf{I}_{\{t\in\tplus\}}.
\]    
Let 
\begin{equation*}
    \begin{split}
        \zeta_{t}(x) &= \zeta_{t}^{(1)}(x)\mathbf{I}_{\{t\in\tminus\}}-\zeta_{t}^{(2)}(x)\mathbf{I}_{\{t\in\tplus\}}, \; S_n(x_{j_m}) = \sum_{\Pi_n}{\zeta_t(x_{j_m})}, \\
        \bm{\lambda}_t(x) &= \left(\zeta_t(x_{t_{j_1}}),\zeta_t(x_{t_{j_2}}),\hdots,\zeta_t(x_{t_{j_k}})\right)', \;
        \mathbf{S}_{n,k} = \left(\sum_{\Pi_n}{S_n(x_{j_1})},\sum_{\Pi_n}{S_n(x_{j_2})},\hdots,\sum_{\Pi_n}{S_n(x_{j_k})}\right)'.
    \end{split}
\end{equation*}

Theorem 3 then follows by applying \Cref{lem:lemma3} to the quadratic characteristic matrix of $\mathbf{S}_{n,k}$. The proof of Theorem 4 follows along similar lines and we omit the details since no technical difficulties are involved.
\end{proof}

\begin{proof}[Proof of Theorem 5]
Assume that we split the data at $t=t_0$. Due to the assumption of increasing sampling frequency, we can consider the number of data points on either side of $t_0$ as $n\rightarrow\infty$. Let $\widetilde\mu(x\mid t_0)=\mu_{t_0^-}(x)-\mu_{t_0^+}(x)$, where $\mu_{t_0^-}(\cdot)$ and $\mu_{t_0^+}(\cdot)$ respectively denote the true conditional mean function before and after the splitting point $t_0$. We denote by $\widehat{\widetilde\mu}(x\mid t_0)$ the Nadaraya Watson kernel estimate of $\widetilde\mu(x\mid t_0)$. Clearly, under relevant bandwidth conditions, as $n\rightarrow\infty$, $\widehat{\widetilde\mu}(x\mid t_0)\xrightarrow{P}\widetilde\mu(x\mid t_0)$
for any $x\in\mathcal{X}$ and $t\in\T$. Now, note that $\widetilde\mu:\mathcal{X}\times\mathcal{T}\rightarrow\mathbb{R}$. Define a map $C:\T\rightarrow\mathbf{C}(\mathcal{X})$, where $\mathbf{C}(A)$ for any set $A$ denotes the set of non-null compact subsets of $A$, as $C(t_j)=x_{t_j}$ where $x_{t_j}\in\Pi_n,t_j\in\T_n$. Since $C$ is a compact-valued map as $n\rightarrow\infty$, then using Berge's maximum theorem we can write $\widehat\mu_{\text{cp}}(t_0)\rightarrow\mu_{\text{cp}}(t_0)$. The consistency of $\widehat\tau_0^\mu$ thus is a consequence of the uniqueness of the maximum of $\widehat{\widetilde\mu}$ in its second argument. One can follow a similar approach for proving the consistency of $\widehat\tau_0^\sigma$.
\end{proof}

\section{Simulation studies}

We start by generating the covariate series $\{X_t\}$ from different data generating processes (DGP), the response $\{Y_t\}$ is then obtained through certain forms of conditional mean and variance functions, subject to different forms of normal, moderate and heavy tailed noise. The considered DGPs, noise distributions and different forms of the conditional mean and variance functions (depending on the number of structural breaks in the data) is described in \Cref{tab:all-simulation-settings}. To detect the power of the test and for assessing the accuracy of the proposed binary segmentation algorithm, structural breaks are randomly introduced in the entire horizon and the segments mentioned in \Cref{tab:all-simulation-settings} refer to different parts generated because of the break-points. 
\begin{table}[!ht]
    \centering
    \caption{Specifications of various components of the simulation settings.}
    \label{tab:all-simulation-settings}
    {
    \begin{tabular}{lll}
    \toprule
    \textbf{Segment} & \textbf{Conditional mean} $\bm{\mu(x)}$ & \textbf{Conditional variance} $\bm{\sigma^2(x)}$ \\
    \midrule
    Segment 1 & $0.5+0.2x$ & $1$ \\
    Segment 2 & $0.1+0.3x^2+0.1x^3+0.2x^4$ & $x^2$ \\
    Segment 3 & $\log(0.4+0.1x^2)$ & $0.1+0.4x^2$ \\
    Segment 4 & $\exp{(0.01x)}$ & $0.5+(0.8+x)^4$ \\
    Segment 5 & $0.9\sin{(x)}$ & $\log(1+0.4x^2)$ \\
    \midrule
    \textbf{DGP} & \textbf{Model Structure} & \textbf{True Parameter Values} \\
    \midrule
    White noise & \( y_t \sim \mathcal{N}(\mu, \sigma^2) + \epsilon_t \) & \( \mu = 0, \sigma = 1 \) \\ 
    ARMA-GARCH & $ y_t = \mu+ \phi_1 y_{t-1}  + \epsilon_t +\theta_1\epsilon_{t-1},$ & \( \mu=0, \phi_1 = 0.5, \theta_1 = -0.4 \) \\
    & $\epsilon_t\sim\gauss(0,\sigma_t^2), \sigma_t^2=\omega+\alpha_1\epsilon_{t-1}^2+\beta_1\sigma_{t-1}^2$ & $\omega=0.1,\alpha_1=0.1,\beta_1=0.8$ \\
    TAR & \(y_t = \begin{cases} \phi_{11} y_{t-1} + \phi_{12} y_{t-2} + \epsilon_t, y_{t-1} \leqslant 0 \\
    \phi_{21}y_{t-1} + \phi_{22} y_{t-2} + \epsilon_t, y_{t-1} > 0 \end{cases} \) & $\begin{array}{l}
    \phi_{11} = 0.6, \phi_{12} = 0.3, \\
    \phi_{21} = -0.6, \phi_{22} = 0.4 
    \end{array}$  \\ 
    \midrule
    \textbf{Noise process} & \textbf{Noise distribution} & \textbf{Parameter specification}  \\
    \midrule
    Normal & $\gauss(\mu,\sigma^2)$ & $\mu=1, \sigma=1$ \\ 
    Student's t & $t_{\nu}$ & $\nu=10$ \\ 
    Power law & $f(x)=x_0\alpha x^{1-\alpha}$ & $x_0=1,\alpha=0.6$ \\
    \bottomrule
    \end{tabular}}
\end{table}

It is pertinent to discuss the choices of the conditional mean and variance structures used in this simulation study. The mean function $\mu(x)$ incorporates a polynomial regression function of different degrees in segments 1 and 2. Segment 3 employs a log-linear regression, frequently used in econometrics to capture diminishing returns and stabilized volatility in financial data. Segment 4 features an exponential regression, widely adopted in finance for modeling compound interest and in biological contexts for growth rates. Finally, in segment 5 we utilize a trigonometric regression model, suitable for modeling seasonal behaviors common in meteorology, economics, and biology. For the variance function $\sigma^2(x)$, the first segment assumes homoskedasticity, while the other specifications are well-established in the literature for addressing heteroskedasticity \citep{palm19967}. Particularly, segments 2 and 3 consider standard GARCH  structures, segment 4 employs a power GARCH framework, and segment 5 utilizes an exponential GARCH approach.

In \Cref{tab:size_power_sigma} we present the simulation performance of the test of changepoints in conditional variance. We observe that the proposed test shows fairly well empirical size control in all considered data generating processes, sample sizes and innovation distributions, with rejection frequencies under the null very close to the nominal level. The test has reasonably high power for a sample size of ($500$) ranging from about $0.56$ to $0.89$ and increases consistently with increasing sample size. The TAR model provides the highest power, being above $0.90$ for large sample sizes. We show that the procedure performs well under Gaussian, Student-$t$ and power-law innovations, so that the performance is not sensitive with respect to the type of the underlying noise distribution. 

\begin{table}[!ht]
\centering
\caption{Performance of the test when there is a single structural break in conditional variance.}
\label{tab:size_power_sigma}
    \begin{tabular}{clcccccc}
    \toprule
     & & \multicolumn{2}{c}{Noise: $\gauss(0,1)$} & \multicolumn{2}{c}{Noise: $t_{10}$} & \multicolumn{2}{c}{Noise: Power law} \\
     Sample size & DGP & Size & Power & Size & Power & Size & Power \\
     \midrule
     500 & White Noise & 0.00 & 0.65 & 0.01 & 0.61 & 0.00 & 0.71 \\
     & ARMA-GARCH & 0.00 & 0.56 & 0.01 & 0.61 & 0.00 & 0.68 \\
     & TAR & 0.01 & 0.81 & 0.01 & 0.73 & 0.03 & 0.89 \\
     \midrule
     1000 & White Noise & 0.00 & 0.66 & 0.01 & 0.69 & 0.01 & 0.78\\
     & ARMA-GARCH & 0.00 & 0.60 & 0.01 & 0.75 & 0.00 & 0.82  \\
     & TAR & 0.00 & 0.88 & 0.03 & 0.88 & 0.01 & 0.91 \\ \midrule
     2000 & White Noise & 0.01 &0.78 & 0.01& 0.72 & 0.00 & 0.86\\
     & ARMA-GARCH & 0.00 & 0.74 & 0.02  & 0.73 &0.03  & 0.84\\
     & TAR & 0.05 & 0.92 &0.04 &0.93 &0.05 &0.97 \\
     \bottomrule
    \end{tabular}
 \end{table}

Next, in \Cref{tab:size_power_mu_sigma}, the performance of the test of structural stability in both mean and variance which, as discussed in Section 3.2 of the main manuscript, is conducted using a Holm-Bonferroni correction. The simultaneous test has around 50\% to 70\% power in smaller sample sizes, working comparatively well in the case of heavy-tailed noises. As the sample size increases the power of the test gets better across all combinations, consistently reaching values in the range of 80\% to 90\%.

\begin{table}[!ht]
\centering
\caption{Performance of the test when there is a single structural break in conditional mean or variance.}
\label{tab:size_power_mu_sigma}
    \begin{tabular}{clcccccc}
    \toprule
     & & \multicolumn{2}{c}{Noise: $\gauss(0,1)$} & \multicolumn{2}{c}{Noise: $t_{10}$} & \multicolumn{2}{c}{Noise: Power law} \\
     Sample size & DGP & Size & Power & Size & Power & Size & Power \\
     \midrule
     500 & White Noise & 0.00 & 0.66 & 0.01 & 0.52 & 0.01 & 0.72 \\
     & ARMA-GARCH & 0.02 & 0.44 & 0.00 & 0.52 & 0.01 & 0.62 \\
     & TAR & 0.00 & 0.72 & 0.01 & 0.68 & 0.00 & 0.74 \\
     \midrule
     1000 & White Noise & 0.00 & 0.78 & 0.00 & 0.70 & 0.01 & 0.74\\
     & ARMA-GARCH & 0.01 & 0.70 & 0.02 & 0.68 & 0.00 & 0.60  \\
     & TAR & 0.01 & 0.80 & 0.00 & 0.64 & 0.00 & 0.80 \\ \midrule
     2000 & White Noise & 0.01 &0.84 & 0.01& 0.80 & 0.00 & 0.82\\
     & ARMA-GARCH & 0.00 & 0.90 & 0.01  & 0.92 &0.00  & 0.88\\
     & TAR & 0.02 & 0.85&0.00 &0.87 &0.01 &0.91 \\
     \bottomrule
    \end{tabular}
\end{table}

Let us move on to the evaluation of the performance of the structural break detection algorithm. Here, we introduce a randomly generated number of breaks in the conditional mean and/or variance functions, placing these breaks at random locations within the data. We ensure that there is a minimum gap of $100$ days between any two breaks to maintain distinct separations and the maximum number of breaks for sample sizes of up to $2000$ is $4$. The effectiveness of the detection process is assessed using several metric. At first we calculate a `deviation' metric, which quantifies the average distance between the detected breaks and the nearest actual structural break. Mathematically, if $\left\{\text{CP}_{1},\text{CP}_2,\hdots,\text{CP}_m\right\}$ are the true structural breaks, then we define the average minimum deviation measure (AMD) as the mean of
\begin{equation*}
    \mathrm{\text{MD}} = \sum_{r=1}^{m'}{\min_{1 \leqslant i \leqslant m}\left\{\abs{\widehat\tau_r-\text{CP}_m}\right\}},
\end{equation*}
taken over all repetitions, where $\left\{\widehat\tau_1,\widehat\tau_2,\hdots,\widehat\tau_{m'}\right\}$ are the detected structural breaks in the data. This average deviation measure, computed over $50$ iterations, provides a quantitative assessment of the accuracy of our detection algorithm by comparing the detected breaks to the true breaks, highlighting the precision and reliability of the procedure. We further compute the average error in terms of number of detected structural breaks, denoted as `ADN' (Average Deviation in Numbers) as the average of 
\begin{equation*}
    \mathrm{DN} = \abs{m-m'}.
\end{equation*}

For the sake of brevity we only show the performance of the detection procedure in the presence of structural breaks in both conditional mean and variance, presented in \Cref{tab:detection_mu_sigma}. For comparison, we also include the accuracy of detected structural breaks using the nonparametric PELT algorithm \citep{haynes2017computationally}. We observe that the nonparametric PELT algorithm severely overestimates the number of structural breaks present in the data, especially for larger sample sizes and more complex data structures. While it helps the method achieve better accuracy in terms of the minimum deviation metric in few cases, it is practically not a suitable approach. Our method, in contrast, estimates at most one additional structural break on average while typically retaining a deviation of less than $\pm 70$ data points from the true structural break. Only exception is when the covariate series comes from an ARMA-GARCH process and there are higher number of structural breaks in the main data. In that scenario, the error in deviation from the true break-point is in the range of 100 to 150.

\begin{table}[ht]
\centering
\caption{Performance of binary segmentation and nonparametric PELT algorithms where there exists a random number of structural breaks in either conditional mean and/or variance}
\label{tab:detection_mu_sigma}
{\footnotesize
\begin{tabular}{@{}ccccccc@{}}
\toprule
\textbf{Sample Size} & \textbf{DGP} & \textbf{Noise} & \multicolumn{2}{c}{\textbf{AMD}} & \multicolumn{2}{c}{\textbf{ADN}} \\ 
                    &               &                & \textbf{CPFind} & \textbf{PELT} & \textbf{CPFind} & \textbf{PELT} \\ \midrule
500                 & White Noise   & $\gauss(0,1)$           & 53.18          & 24.49         & 0.32            & 0.97          \\
500                 & White Noise   & $t_{10}$              & 57.88           & 21.96         & 0.28            & 0.96          \\
500                 & White Noise   & Power law      & 17.30           & 29.46         & 0.24            & 1.03          \\ \hline
500                 & ARMA-GARCH    & $\gauss(0,1)$           & 76.52          & 32.21         & 0.42            & 1.10          \\
500                 & ARMA-GARCH    & $t_{10}$              & 82.98          & 29.21        & 0.48            & 1.12          \\
500                 & ARMA-GARCH    & Power law      & 41.24          & 29.27        & 0.22            & 1.40          \\ \hline
500                 & TAR           & $\gauss(0,1)$           & 61.20          & 30.02         & 0.26            & 1.21          \\
500                 & TAR           & $t_{10}$              & 65.50          & 29.14         & 0.22            & 1.38          \\
500                 & TAR           & Power law      & 24.06          & 38.02         & 0.20            & 1.38          \\ \hline
1000 & White Noise   & $\gauss(0,1)$           & 42.34          & 34.90         & 0.46            & 1.46          \\
1000 &  White Noise   & $t_{10}$              & 68.76           & 95.76         & 0.52            & 1.34          \\
1000 &  White Noise   & Power law      & 17.54           & 99.76         & 0.68            & 2.22          \\ \hline
1000 &  ARMA-GARCH    & $\gauss(0,1)$           & 104.46          & 55.40         & 0.60            & 2.98          \\
1000 &  ARMA-GARCH    & $t_{10}$              & 169.84          & 108.21        & 0.80            & 2.34          \\
1000 &  ARMA-GARCH    & Power law      & 48.28          & 132.91        & 0.38            & 3.60          \\ \hline
1000 &  TAR           & $\gauss(0,1)$           & 53.06          & 83.65         & 0.46            & 3.76          \\
1000 & TAR           & $t_{10}$              & 64.28          & 85.65         & 0.44            & 3.38          \\
1000 & TAR           & Power law      & 30.82         & 70.63         & 0.66            & 3.68          \\  \bottomrule
2000 & White Noise  & $\gauss(0,1)$           & 57.14       & 67.30         & 1.06 & 1.76  \\
2000 & White Noise  & $t_{10}$             & 53.73       & 56.02        & 1.46 & 1.68  \\
2000 & White Noise  & Power law      & 24.94        & 91.73        & 1.71 & 1.42  \\ \hline
2000 & ARMA-GARCH   & $\gauss(0,1)$           & 146.88       & 62.57        & 1.10 & 3.58  \\
2000 & ARMA-GARCH   & $t_{10}$            & 120.00       & 109.77        & 0.69 & 3.54  \\
2000 & ARMA-GARCH   & Power law      & 115.65        & 113.92        & 0.82 & 3.46  \\ \hline
2000 & TAR          & $\gauss(0,1)$           & 72.00        & 68.38        & 1.62 & 5.34  \\
2000 & TAR          & $t_{10}$              & 47.31       & 76.23        & 1.88 & 5.26  \\
2000 & TAR          & Power law      & 21.24        & 68.02        & 1.70 & 4.96  \\ 
\bottomrule
\end{tabular}
}
\end{table}

\section{Additional analysis of Bitcoin data}

We begin with a comparative analysis of the segments identified by the structural break tests using key descriptive measures. Here, ``JB-Test'' denotes Jarque-Bera test of normality, ``LB-Test'' denotes the Ljung-Box test for autocorrelation, ``LM-Test'' denotes the Lagrange multiplier test for heteroskedasticity and ``ADF-Test'' denotes the augmented Dickey-Fuller test for stationarity. All tests are conducted at 5\% level of significance. 

\begin{table}[!ht]
\centering
\footnotesize
\caption{Descriptive statistics for log-price, returns and GNIS, for the three segments generated by the detected structural breaks. Segment 1 refers to 1 January 2020 to 7 March 2021, Segment 2 is 8 March 2021 to 8 July 2023, and Segment 3 is from 8 July 2023 until 4 September 2024.}
\setlength{\tabcolsep}{5pt} 
\resizebox{\textwidth}{!}{ 
\begin{tabular}{@{}lccccccccc@{}}
\toprule
\textbf{Entire data}               & \textbf{Mean (SD)} & \textbf{Range} & \textbf{Quartiles} & \textbf{Skewness} & \textbf{Kurtosis} & \textbf{JB Test} & \textbf{LB Test} & \textbf{LM Test} & \textbf{ADF Test} \\ \midrule
\textbf{Log-Price}      & 10.24 (0.64)       & (8.50, 11.20)  & (9.86, 10.32, 10.75) & $-0.57$          & 2.37              & 0.00            & 0.00             & 1.00             & 0.76              \\
\textbf{Log-Returns}    & 0.00 (0.03)        & ($-0.32$, 0.16) & ($-0.02$, 0.00, 0.02) & $-0.60$         & 10.60             & 0.00            & 0.41             & 0.00             & 0.01              \\
\textbf{GNIS}           & 28.72 (16.32)      & (10, 100)       & (17.82, 22.71, 35.00) & 1.79            & 6.51              & 0.00            & 0.00             & 1.00             & 0.02              \\
\midrule
\textbf{Segment 1}               & \textbf{Mean (SD)} & \textbf{Range} & \textbf{Quartiles} & \textbf{Skewness} & \textbf{Kurtosis} & \textbf{JB Test} & \textbf{LB Test} & \textbf{LM Test} & \textbf{ADF Test} \\ \midrule
\textbf{Log-Price}      & 9.46 (0.58)       & (8.50, 10.95)  & (9.12, 9.26, 9.78) & 1.07          & 3.10              & 0.00            & 0.00             & 1.00             & 0.79              \\
\textbf{Log-Returns}    & 0.00 (0.00)        & ($-0.36$, 0.02) & (0.00, 0.00, 0.00) & $-1.32$         & 18.11             & 0.00            & 0.91             & 0.00             & 0.01              \\
\textbf{GNIS}           & 27.67 (22.99)      & (10, 98)       & (14.00, 16.93, 28.96) & 1.75            & 4.81              & 0.00            & 0.00             & 1.00             & 0.56              \\
\midrule
\textbf{Segment 2}               & \textbf{Mean (SD)} & \textbf{Range} & \textbf{Quartiles} & \textbf{Skewness} & \textbf{Kurtosis} & \textbf{JB Test} & \textbf{LB Test} & \textbf{LM Test} & \textbf{ADF Test} \\ \midrule
\textbf{Log-Price}      & 10.38 (0.40)       & (9.66, 11.12)  & (10.03, 10.36, 10.72) & $-0.01$          & 1.80              & 0.00            & 0.00             & 1.00             & 0.78              \\
\textbf{Log-Returns}    & 0.00 (0.00)        & ($-0.02$, 0.01) & (0.00, 0.00, 0.01) & $-0.37$         & 6.90            & 0.00            & 0.30             & 0.00             & 0.01              \\
\textbf{GNIS}           & 32.16 (14.35)      & (15, 100)       & (20.57, 29.71, 39.00) & 1.32            & 5.17              & 0.00            & 0.00             & 1.00             & 0.01              \\
\midrule
\textbf{Segment 3}               & \textbf{Mean (SD)} & \textbf{Range} & \textbf{Quartiles} & \textbf{Skewness} & \textbf{Kurtosis} & \textbf{JB Test} & \textbf{LB Test} & \textbf{LM Test} & \textbf{ADF Test} \\ \midrule
\textbf{Log-Price}      & 10.79 (0.35)       & (10.13, 11.20)  & (10.35, 10.72, 11.06) & $-0.31$          & 1.58              & 0.00            & 0.00             & 1.00             & 0.92              \\
\textbf{Log-Returns}    & 0.00 (0.00)        & ($-0.01$, 0.01) & (0.00, 0.00, 0.00) & 0.06         & 4.20             & 0.00            & 0.41             & 0.00             & 0.01              \\
\textbf{GNIS}           & 22.92 (8.28)      & (12, 54)       & (18.00, 20.71, 25.14) & 1.64            & 5.83              & 0.00            & 0.00             & 1.00             & 0.67              \\
\bottomrule
\end{tabular}
}
\label{tab:decriptive_stats}
\end{table}

The descriptive statistics and these tests provide a clear picture of the market behavior across the three segments. The first one, i.e., 1 January 2020 to 7 March 2021, is characterized by relatively lower log-prices, high volatility in log-returns (as indicated by the high kurtosis), and the most dispersed GNIS values. The positive skewness of log-prices and negative skewness of log-returns highlight significant price increases and occasional sharp decreases during this period. In this segment, the JB test reveals significant deviations from normality, particularly in log-returns, which exhibit high kurtosis. The LB test indicates autocorrelation in log-returns, while the LM test suggests heteroskedasticity, though only in log-returns. The ADF test confirms the stationarity of log-returns, but log-prices and GNIS are borderline non-stationary. 

Segment 2  (8 March 2021 to 8 July 2023) shows a rise in the mean log-price, reduced volatility in log-returns (lower kurtosis), and increased GNIS, suggesting greater market stability and media attention compared to Segment 1. Log-prices have a nearly symmetrical distribution, indicating fewer extreme price movements. The JB test in this segment shows improvements in normality for log-prices and GNIS, though log-returns still deviate. Autocorrelation weakens during this period, as seen in the LB test, and there is no evidence of heteroskedasticity in any variables, reflecting greater stability. However, the ADF test still indicates trends in log-prices and GNIS, while log-returns remain stationary.

Finally, Segment 3 (8 July 2023 to 4 September 2024) marks the highest average log-prices with relatively lower GNIS compared to Segment 2, showing a moderate decrease in media attention. The volatility in log-returns remains low, while skewness and kurtosis suggest more stabilized market behavior, almost approaching the normal distribution. While log-returns continue to exhibit non-normality, the deviations are less pronounced, and there is no autocorrelation or heteroskedasticity detected. The volatility in GNIS diminishes, and stationarity is confirmed for log-returns, though trends persist in log-prices and GNIS.

As a last piece of empirical illustration, \Cref{fig:before_after_cb} presents the confidence bands for the disparity in mean regression functions caused by the two structural breaks, as well as the conditional variance function in the entire time horizon. 

\begin{figure}[!ht]
\centering
\captionsetup{justification=centering}
\centerline{\includegraphics[width=0.65\textwidth,keepaspectratio]{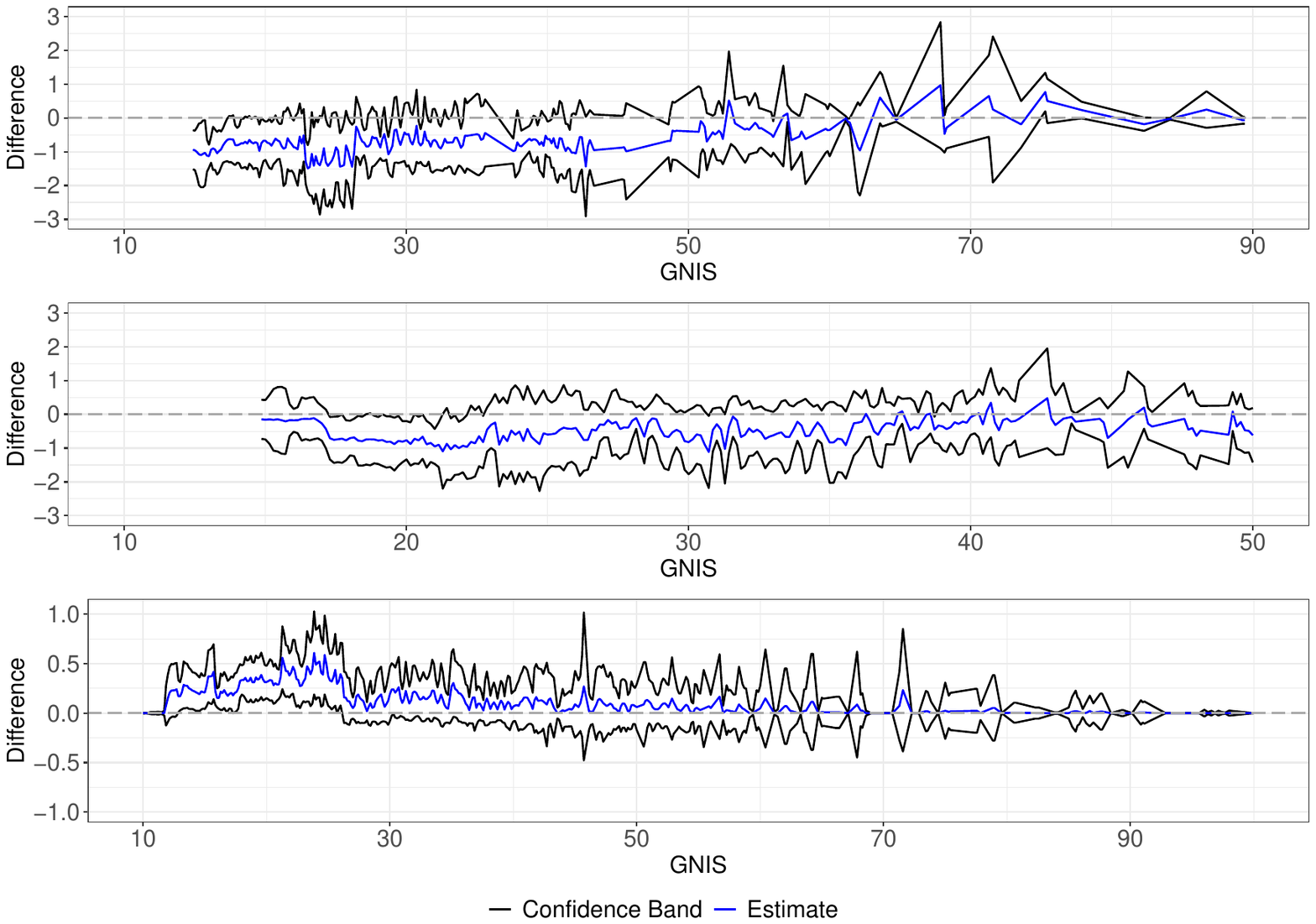}}
\caption{Confidence intervals of: (top) the difference in $\mu(x)$ before and after the first structural break, (middle) the difference in $\mu(x)$ before and after the second structural break, (bottom) $\sigma^2(x)$.}
\label{fig:before_after_cb}
\end{figure}

The top panel of \Cref{fig:before_after_cb} indicates a slightly increasing trend in the effect of GNIS on the mean function. It further highlights that when news attention is low, there are greater fluctuations in the difference between the average log-price behavior on either side of the first structural break. However, as news attention increases, the confidence band narrows and the difference in the effect on mean log-price diminishes. It implies that lesser news attention has less impact on the average price in the first segment as compared to the second, whereas heightened news attention brings more certainty in the mean price movement surrounding the break. The middle panel of the same plot illustrates a uniformly wider confidence band of the disparity in the mean function, implying that the estimate of the difference has more variability around the second structural break.

The bottom panel of the same figure examines the conditional variance function. The confidence band for the conditional variance is uniformly much narrower than the bands for the disparity in mean level. This shows that the conditional variance has been more consistent throughout the time horizon, further justifying the absence of a structural break. Additionally, we observe a similar trend as in the conditional mean: lower news attention correlates with greater volatility, as seen by wider fluctuations in the variance. Conversely, with increased news coverage, the volatility stabilizes, demonstrating that as the media focuses more on Bitcoin, the volatility in its price becomes more predictable and less prone to sudden changes. Overall, the figure underscores how news sentiment impacts both the mean and variance of Bitcoin prices, with greater news attention leading to more stable price behavior. This analysis underscores the significant role of news attention in shaping both price behavior and volatility patterns in the Bitcoin market.

\end{document}